\newcommand*{\LONG}{}%

\ifdefined\LONG
\documentclass{article}
\usepackage{amsfonts,amsmath,amssymb,latexsym,comment,amsthm}

\usepackage[letterpaper]{geometry}
\usepackage{fullpage}

\usepackage{mathrsfs}
\usepackage{times}

\else 
\documentclass{sig-alternate-2013}
\usepackage{amsfonts,amsmath,amssymb,latexsym,comment }

\newfont{\mycrnotice}{ptmr8t at 7pt}
\newfont{\myconfname}{ptmri8t at 7pt}

\permission{Permission to make digital or hard copies of all or part of this work for personal or classroom use is granted without fee provided that copies are not made or distributed for profit or commercial advantage and that copies bear this notice and the full citation on the first page. Copyrights for components of this work owned by others than ACM must be honored. Abstracting with credit is permitted. To copy otherwise, or republish, to post on servers or to redistribute to lists, requires prior specific permission and/or a fee. Request permissions from permissions@acm.org.}
\conferenceinfo{STOC'15,}{June 14--17, 2015, Portland, Oregon, USA.}
\copyrightetc{Copyright \copyright~2015 ACM \the\acmcopyr}
\crdata{978-1-4503-3536-2/15/06\ ...\$15.00.\\
	http://dx.doi.org/10.1145/2746539.2746581}
\fi

\usepackage{qtree}

\usepackage{xspace}
\usepackage{cases}
\ifx\pdftexversion\undefined
\usepackage[colorlinks,linkcolor=black,filecolor=black,citecolor=black,urlcolor=black,pdfstartview=FitH]{hyperref}
\else
 \usepackage[colorlinks,linkcolor=blue,filecolor=blue,citecolor=blue,urlcolor=blue,pdfstartview=FitH]{hyperref}
\fi
\usepackage[lined,ruled,boxed,commentsnumbered,noresetcount]{algorithm2e}
\usepackage{enumitem}

\usepackage{pifont}

\def\squarebox#1{\hbox to #1{\hfill\vbox to #1{\vfill}}}

\newtheorem{theorem}{Theorem}

\newtheorem{lemma}{Lemma}

\newtheorem{property}{Property}
\newtheorem{definition}{Definition}

\newtheorem{claim}{Claim}
\newtheorem{corollary}{Corollary}

\newcommand{\namedref}[2]{\hyperref[#2]{#1~\ref*{#2}}}
\newcommand{\sectionref}[1]{\namedref{Section}{#1}}

\newcommand{\propertyref}[1]{\namedref{Property}{#1}}

\newcommand{\theoremref}[1]{\namedref{Theorem}{#1}}

\newcommand{\figureref}[1]{\namedref{Figure}{#1}}

\newcommand{\claimref}[1]{\namedref{Claim}{#1}}
\newcommand{\lemmaref}[1]{\namedref{Lemma}{#1}}

\newcommand{\corollaryref}[1]{\namedref{Corollary}{#1}}

\newcommand{\etal}{et~al.\xspace}

\newcommand{\false}{\mathit{false}}
\newcommand{\true}{\mathit{true}}

\newcommand{\ie}{\emph{i.e.,\xspace}}

\def\beginsmall#1{\vspace{-\parskip}\begin{#1}\itemsep-\parskip}
\def\endsmall#1{\end{#1}\vspace{-\parskip}}

\newcommand{\RT}{\mathcal{RT}\!}
\newcommand{\PT}{\mathcal{PT}\!}

\newcommand{\CT}{\mathcal{CT}\!}

\newcommand{\IT}{\mathcal{IT}\!}

\newcommand{\F}{\mathcal{F}}
\newcommand{\FA}{\mathcal{F\!A}}
\newcommand{\D}{\mathcal{D\!}}
\newcommand{\Df}{\mathcal{D}_{\!\phi}}
\newcommand{\bad}{\mbox{\sc bad}\xspace}

\newcommand{\tb}{\makebox[0.6cm]{}}

\newcommand{\hide}[1]{}

\newcommand{\commentout}[1]{}

\newcommand{\itrule}{\textsc{it-to-rt rule}\xspace} 
\newcommand{\earlyitrule}{\textsc{early\!\_it-to-rt rule}\xspace} 
\newcommand{\strongitrule}{\textsc{strong\_it-to-rt rule}\xspace} 
\newcommand{\decayrule}{\textsc{decay rule}\xspace} 
\newcommand{\gcrule}{\textsc{resolve rule}\xspace} 
\newcommand{\rgcrule}{\textsc{relaxed rule}\xspace} 
\newcommand{\srule}{\textsc{special-bot rule}\xspace} 
\newcommand{\srootrule}{\textsc{special-root-bot rule}\xspace} 
 \newcommand{\lastroundrule}{\textsc{round $\phi+1$ rule}\xspace} 

\begin{document}


\title{Byzantine Agreement with Optimal Early Stopping, Optimal Resilience and Polynomial Complexity}
\author{
Ittai Abraham\thanks{Part of the work was done
	at Microsoft Research Silicon Valley.}\\
VMware Research
\\Palo Alto, CA, USA\\
\texttt{iabraham@vmware.com}
\and
Danny Dolev\thanks{Part of the work was done
while the author visited Microsoft Research Silicon Valley.
Danny Dolev is Incumbent of the Berthold Badler Chair in Computer Science. 
This research project was supported in part by The Israeli Centers of Research Excellence (I-CORE) program, (Center  No. 4/11), and by grant 3/9778 of the Israeli Ministry of Science and Technology.
}\\
Hebrew University of Jerusalem\\
Jerusalem, Israel\\
\texttt{dolev@cs.huji.ac.il}
}

\maketitle


\begin{abstract}
We provide the first protocol that solves Byzantine agreement with optimal early stopping 
($\min\{f+2,t+1\}$ rounds) and optimal resilience ($n>3t$) using polynomial message size and computation.

All previous approaches obtained sub-optimal results and used resolve rules that looked only at the immediate children in the EIG (\emph{Exponential Information Gathering})  tree.
At the heart of our solution are new resolve rules that look at multiple layers of the EIG tree.
\end{abstract}


\section{Introduction}

In 1980 Pease, Shostak and Lamport \cite{PSL80, LPS82} introduced the problem of Byzantine agreement, a fundamental problem in fault-tolerant distributed computing.
In this problem $n$ processes each have some initial value and the goal is to have all correct processes decide on some common value. The network is reliable and synchronous. If all correct processes start with the same initial value then this must be the common decision value, and otherwise the value should either be an initial value of one of the correct processes or some pre-defined default value.\footnote{Other versions of the problem may not restrict to a value on one of the correct processes, if not all initial values are the same, or require agreement on a leader's initial value, which can be reduced to the version we defined.} This should be done in spite of at most $t$ corrupt processes that can behave arbitrarily (called Byzantine processes).
Byzantine agreement abstracts one of the core difficulties in distributed computing and secure multi-party computation --- that of coordinating a joint decision.
Pease \etal \cite{PSL80} prove that Byzantine agreement cannot be solved for $n \le 3t$. Therefore we say that a protocol that solves Byzantine agreement for $n>3t$ has \emph{optimal resilience}. Fisher and Lynch \cite{FL82} prove that any protocol that solves Byzantine agreement must have an execution that runs for $t+1$ rounds.  Dolev \etal \cite{DRS90} prove that any protocol must have executions that run for $\min \{ f+2,t+1\}$ rounds, where $f$ is the actual number of corrupt processes. Therefore we say that a protocol that solves Byzantine agreement with  $\min \{ f+2,t+1\}$ rounds has \emph{optimal early stopping}.

The protocol of  \cite{PSL80} has optimal resilience and optimal worst case $t+1$ rounds. However the message complexity of their protocol is exponential. Following this result, many have studied the question of obtaining a protocol with optimal resilience and optimal worst case rounds that uses only polynomial-sized messages (and computation).

Dolev and Strong~\cite{DS82} obtained the first polynomial protocol with optimal resilience.
The problem of obtaining a protocol with optimal resilience, optimal worst case rounds and polynomial-sized messages turned out to be surprisingly challenging. Building on a long sequence of results, Berman and Garay \cite{BG93-votes} presented a protocol with optimal worst case rounds and polynomial-sized messages for $n>4t$. In an exceptional tour de force, Garay and Moses \cite{GM93, GM98}, presented a protocol for binary-valued Byzantine agreement obtaining optimal resilience, polynomial-sized messages and $\min \{f+5,t+1\}$ rounds. We refer the reader to \cite{GM98} for a detailed and full account of the related work. Recently Kowalski and Most{\'e}faoui \cite {KM03} improved the message complexity to $\tilde{O}(n^3)$ but their solution does not provide early stopping and requires exponential computation.

Worst case running of $t+1$ rounds is the best possible if the protocol is to be resilient to an adversary that controls $t$ processes. However, in executions where the adversary controls only $f<t$ processes, the optimal worst case can be improved to $f+2$ rounds. Berman \etal~\cite{BGP92} were the first to obtain optimal resilience and optimal early stopping (i.e. $\min \{ f+2,t+1\}$ rounds) using exponential size messages. Early stopping is an extremely desirable property in real world replication systems. In fact, agreement in a small number of rounds when $f=0$ is a core advantage of several practical state machine replication protocols (for example \cite{CCL99} and \cite{ZYZ07} focus on optimizing early stopping in the fault free case).

Somewhat surprisingly, after more than 30 years of research on Byzantine Agreement, the problem of obtaining the best of all worlds is still open. There is no protocol with optimal resilience, optimal early stopping and polynomial-sized message. The conference version of \cite{GM98} claimed to have solved this problem but the journal version only proves a $\min \{f+5,t+1\}$ round protocol, then says it is \emph{possible} to obtain a $\min \{ f+3,t+1\}$ round protocol and finally the authors say they  \emph{believe} it should be possible to obtain a $\min \{ f+2,t+1\}$ round protocol. We could not see how to directly extend the approach of  \cite{GM98} to obtain optimal early stopping. The main contribution of this paper is solving this long standing open question
and providing the optimal $\min \{ f+2,t+1\}$ rounds with optimal resilience and polynomial complexity. 
Moreover, our result applies directly for arbitrary 
initial values and not only to binary initial values, as some of the previous results.

Our Byzantine agreement protocol obtains a stronger notion of \emph{multi-valued validity}. If $v \neq \bot$ is the decision value then at least $t+1$ correct processes started with value $v$. 
The multi-valued validity property is crucial in our solution for early stopping with monitors. This property is also more suitable in proving that Byzantine agreement implements an ideal world centralized decider that uses the majority value. We note that several previous solutions (in particular \cite{GM98}) are inherently binary and their extension to multi-valued agreement does not have the stronger multi-valued validity property.

\begin{theorem}\label{thm:full}
	Given $n$ processes, there exists a protocol that solves Byzantine agreement. The protocol is resilient to any Byzantine adversary of size $t<n/3$. For any such adversary, the total number of bits sent by any correct process is polynomial in $n$ and the number of rounds is $\min \{f+2,t+1\}$ where $f$ is the actual size of the adversary.
\end{theorem}

\textbf{Overview of our solution.}
At a high level we follow the framework set by Berman and Garay \cite{BG93-votes}. In this framework, if at a given round all processes seem to behave correctly then the protocol stops quickly thereafter. So if the adversary wants to cause the protocol to continue for many rounds it must have at least one corrupt process behave in a faulty manner in each round. However, behaving in a faulty manner will expose the process and in a few rounds the mis-behaving process will become publicly exposed as corrupt.

This puts the adversary between a rock and a hard place: if too few 
corrupt processes are publicly exposed then the protocol reaches agreement quickly, if too many corrupt process are exposed then a ``monitor" framework (also called ``cloture votes") that runs in the background
causes the protocol to reach agreement in a few rounds. So the only path the adversary can take in order to generate a long execution is to publicly expose exactly one corrupt process each round. In the $t<n/4$ case, this type of adversary behavior keeps the communication polynomial.

For $t<n/3$ a central challenge is that a corrupt process can cause communication to grow in round $i$ but will be  
publicly exposed only in round $i+2$. Naively, such a corrupt process may also cause communication to grow both in round $i$ and $i+1$ and this may cause exponential communication blowup. Garay and Moses \cite{GM98} overcome this challenge by providing a protocol such that, if there are at most two new corrupt processes in round $i$ and no new corrupt process in round $i+1$ then even though they are publicly exposed in round $i+2$ they cannot increase communication in round $i+1$ (also known as preventing ``cross corruption").

At the core of the binary-valued protocol of Garay and Moses is the property that one value can only be decided on even rounds and the other only on odd rounds.  This property seems to raise several unsolved challenges for obtaining optimal early stopping. We could not see how to overcome these challenges and obtain optimal early stopping using this property. Our approach allows values to be fixed in a way that is indifferent to the parity of the round number (and is not restricted to binary values).

Two key properties of our protocol that makes it quite different from all previous protocols. First, the value of a node is determined by the values of its children and grandchildren in the EIG tree (\cite{BDDR92}). Second, if agreement is reached on a node then the value of all its children is changed to be the value of the node. This second property is crucial because otherwise even though a node is fixed there could be disagreement about the value of its child. Since the value of the parent of the fixed node depends on its children and grandchildren, the disagreement on the grandchild may cause disagreement on the parent and this disagreement could propagate to the root.

The decision to change the value of the children when their parent is fixed is non-trivial. Consider the following scenario with a node $\sigma$, child $\sigma p$ and grandchild $\sigma p q$: some correct reach agreement that the value of $\sigma p q$ is $d$, then some correct reach agreement that the value of $\sigma p$ is $d'\ne d$ and hence the value of $\sigma pq$ is changed (colored) to $d'$. So it may happen that some correct decide the value of $\sigma$ based on $\sigma p q$ being fixed on $d$ and some other correct decide the value of $\sigma $ based on $\sigma p q$ being colored to $d'$. Making sure that agreement is reached in all such scenarios requires us to have a relatively complex set of complementary agreement rules. 

To bound the size of the tree by a polynomial size we prove that the adversary is still between a rock and a hard place: roughly speaking there are three cases. If just one new process is publicly exposed in a given round then the tree grows mildly (remains polynomial). If three or more new processes are exposed in the same round then this increases the size of the tree but can happen at most a constant number of times before a monitor process will cause the protocol to stop quickly.

The remaining case is when exactly two new processes are exposed, then a sequence of (possibly zero) rounds where just one new process is exposed in each round, followed by a round where no new process is exposed. This is a generalized version of the ``cross corruption" case of \cite{GM98} where the adversary does not face increased risk of being caught by the monitor process. We prove that in these cases the tree essentially grows mildly (remains polynomial).

In order to deal with this generalized ``cross corruption'' we introduce a special resolve rule ($\srule$) tailored to this scenario. In particular, in some cases we fix the value of a node $\sigma$ to $\bot$ (a special default value) if we detect enough support. This solves the generalized ``cross corruption'' problem but adds significant complications. Recall that when we fix a value to a node then we also fix (color) the children of this node with the same value.

Suppose a process fixes a node $\sigma$ to $\bot$. The risk is that some correct processes may have used a child $\sigma p$ with value $d$ but some other correct process will see $\bot$ for $\sigma p$ (because when $\sigma$ is fixed to $\bot$ we color all its children to $\bot$). Roughly speaking, we overcome this difficulty by having two resolve rule thresholds. The base is the $n-t$ threshold ($\gcrule$, $\itrule$) and the other is with a $n-t-1$ threshold ($\rgcrule$). In essence this $n-t-1$ rule is resilient to 
disagreement on one child node 
(that may occur due to coloring). We then make sure that the $\srule$ can indeed change only one child value. This delicate interplay between the resolve rules is at the core of our new approach.

\textbf{The adversary.}
Given $n>3t$ and $\phi \le t$, as in \cite{GM98}, we will consider a \emph{$(t,\phi)$-adversary} - an adversary that can control up to $\phi$ corrupt processes that behave arbitrarily and at most $t-\phi$ corrupt processes that are always silent (send some default value $\bot$ to all processes every round).
The $(t,\phi)$-adversary will be useful to model executions in which all correct processes have detected beforehand some common set of at least $t-\phi$ corrupt processes  and hence ignore them throughout the protocol. Note that the standard $t$-adversary is just a $(t,t)$-adversary.

\section{The EIG structure and rules}\label{sec:eig}

In this section we define the EIG structure and rules.

Let $N$ be the set of processes, $n = |N|$ and assume that $n>3t$. Let $D$ be a set of possible decision values. We assume some decision $\bot \in D$ is the designated default decision. 

Let $\Sigma_r$ be the set of all sequences  of length $r$ of elements of $N$ without repetition. Let $\Sigma_0=\epsilon$, the empty sequence. Let $\Sigma = \bigcup_{0 \le j \le t+1} \Sigma_j$. An \emph{Exponential Information Gathering} tree (EIG in short) is a tree whose nodes are elements in $\Sigma$ and whose edges connect each node to the node representing its longest proper prefix.  Thus, node $\epsilon$ has $n$ children, and a node from $\Sigma_k$ has exactly $n-k$ children.

We will typically use the Greek letter $\sigma$ to denote a sequence (possibly empty) of labels corresponding to a node in an EIG tree. We use the notation $\sigma q$ to denote the node in the EIG tree that corresponds
to the child of node $\sigma$ that corresponds
to the sequence $\sigma$ concatenated with $q \in N$.
We denote by $\bar\epsilon$ the root node of the tree that corresponds to
the empty sequence. 
Given two sequences $\sigma, \sigma' \in \Sigma$, let $\sigma' \sqsubset \sigma$ denote that $\sigma'$ is a proper prefix of $\sigma$ and $\sigma' \sqsubseteq \sigma$ denote that $\sigma'$ is a prefix of $\sigma$ (potentially $\sigma'=\sigma$).

In the EIG consensus protocol each process maintains a dynamic tree data structure  $\IT$. This data structure maps a set of nodes in $\sigma$ to values in $D$. Intuitively, this tree contains all the information the process has heard so far. Each process $z$ also maintains two global dynamic sets $\F,\FA$. 
The set $\F$ contains processes that $z$ detected as faulty, and $\FA$ contains processes that $z$ knows are detected by all correct processes. The protocol for updating $\F,\FA$ is straightforward: 
\beginsmall{itemize}
\item In each round the processes exchange their $\F$ 
lists
and update their $\F$ and $\FA$ sets once a faulty process appears in $t+1$ or $2t+1$ lists, respectively. 
\item When a process is detected as faulty every correct process masks
its future messages to $\bot$. 
\endsmall{itemize}

The basic EIG 
protocol will be invoked repeatedly, and several copies of the EIG 
protocol may be running concurrently.  The accumulated set of faulty processes will be used across all copies (the rest of the variables and data structures are local to each EIG invocation).  Therefore, we assume that when the protocol is invoked the following property holds:

\begin{property}\label{prop:init-fail}
	When the protocol is invoked, no correct process appears in the faulty sets of any other correct process. Moreover, $\FA_p\subseteq\F_p$ and $\FA_p\subseteq\F_q$ for any two correct processes $p$ and $q$, 
\end{property}

Each  invocation of the EIG protocol is tagged with a parameter $\phi$, known to all  processes. An EIG protocol with parameter $\phi$, will run for at most $\phi+1$ rounds. 
At the beginning of the agreement protocol the faulty sets are empty at all correct processes and the EIG protocol with parameter $\phi=t$ is executed. Each additional invocation of the EIG protocol is with a smaller value of $\phi.$
In the non-trivial case, when the EIG protocol with parameter $\phi$ is invoked then $|\bigcap_i \FA_i |\ge t-\phi.$ There will be one exception to this assumption, and it is handled in \lemmaref{lem:input-agree}.  Thus, other than in that specific case, 
it is assumed that we have a $(t,\phi)$-adversary during the execution of the EIG protocol with parameter $\phi$.

The basic EIG protocol for a correct process $z$ with initial value $d_z\in D$ is very simple:
\beginsmall{enumerate}
\item \textbf{Init:} Set $\IT(\bar\epsilon) := d_z$, so $\IT(\bar\epsilon)$ is set to be the initial value.

\item \textbf{Send:} in each round $r$, $1\le r\le \phi+1$, for every $\sigma \in \IT \cap \Sigma_{r-1}$, such that $z \notin \sigma$, send the message $\langle \sigma, z, \IT(\sigma) \rangle$ to every process.

\item \textbf{Receive set:} in each round $r$, let $\mathcal{S}_r:=\{ \sigma x \in \Sigma_r\}$.

\item \textbf{Receive rule:} in each round $r$, for all $\sigma x \in \mathcal{S}_r$ set \\
$\IT(\sigma x) :=
\begin{cases}
\bot &\mbox{if }x \in \F \\
d & \mbox{if } x \not\in\F \mbox{ sent }  \langle \sigma, x, d \rangle \mbox{ and } d\in D;\\
\IT(\sigma) & 
\mbox{otherwise.}
\end{cases}$

\endsmall{enumerate}

\textit{Note:} assigning of $\IT(\sigma x) := \IT(\sigma)$ when $x \notin \F$ is crucial 
for the case where
$x$ is correct and has halted in the previous round. Thus, if a process is silent but is not detected (possibly because it has halted due to early stopping) $z$ assigns it the value it heard in the previous round. 

We use a second dynamic EIG tree data structure  $\RT$. Intuitively, if a process puts a value in a node of this tree then, essentially, all correct processes will put the same value in the same node in at most 2 more rounds. Processes use several rules to close branches of the $\IT$ tree whose value in $\RT$ is already determined by all. We present later the rules for closing branches of the $\IT$ tree.
To handle this,  we modify lines 2 and 3 as described below (and keep lines 1 and 4 as above).
\beginsmall{enumerate}

\item [2.] \textbf{Send:} in each round $r$, $1\le r\le \phi+1$, for every $\sigma \in \IT \cap \Sigma_{r-1}$, such that
$z \mbox{$\notin \sigma$}$, and the branch $\sigma$ is not closed
send the message $\langle \sigma, z, \IT(\sigma) \rangle$ to every process.

\item [3.] \textbf{Receive set:} in each round $r$, let $\mathcal{S}_r = \{\sigma x \in \Sigma_r \mid \mbox{branch }\sigma x \mbox{ is not closed} \}$.

\endsmall{enumerate}

Informally, $\IT_z(\sigma p) =d$ 
(where $\IT_z$ denotes the $\IT$ tree at process $z$)
 indicates that  process $z$ received a message from process $p$ that said that his value for $\sigma$ was $d$. $\RT_z(\sigma p)=d$ indicates, essentially, that  process $z$ knows that every correct process $x$ will agree and have $d \in \RT_x(\sigma p)$ in at most two more rounds.

Observe that we record in the EIG tree only information from sequences of nodes that do not contain repetition, therefore, not every message a process receives will be recorded.

At the end of each round, we  apply the rules below to determine whether to assign values to nodes in $\RT$, assigning that value in $\RT$ is called {\em resolving} the node.


\subsection{The Resolve Rules}\label{sub:EGC}

A key feature of our algorithm is that whenever we put a value into $\RT(\sigma)$ we also color (assign) all the descendants of $\sigma$ in $\RT$ with the same value. Observe that this means we may color a node $\sigma w$ in $\RT$ to $d$ even if $w$ is correct and sent $d' \neq d$ to all other correct processes.

\textbf{Rules for IT-to-RT resolve:}
The following definitions and rules cause a node to be resolved based on information in $\IT$.

\beginsmall{enumerate}

\item If  $\IT(\sigma w)=d$ \textbf{then} we say: (1) $w$ is a voter of $(\sigma,w,d)$; (2) $w$ is confirmed on $(\sigma,w,d)$; (3) For all $v \in N\setminus\{\sigma\}$, $w$ is a supporter of $v$ on $(\sigma,w,d)$.

\textit{Note:} the reason that we count $w$ as a voter, as confirmed and as a supporter for all its echoers is that due to the EIG structure $w$ does not appear  in the subtree of $\sigma w$.\\

\item If $\IT(\sigma w v)=d$ , \textbf{then} we say that
$v$ is a supporter of $v$ for $(\sigma,w,d)$.

\textit{Note:} again we need $v$ to be a supporter of itself because of the EIG structure.\\

\item If  $\IT(\sigma w v u) =d$  \textbf{then} we say that $u$ is a supporter of $v$ for $(\sigma,w,d)$.\\

\item If there is a set $|U|=n-t$, such that for each $u' \in U$, $u'$ is a supporter of $v$ on $(\sigma,w,d)$ \textbf{then} we say that $v$ is confirmed on $(\sigma,w,d)$.

\textit{Note:} if $\sigma$ contains no correct and $w$ is correct, then any correct child $v$ (of $\sigma w$) will indeed have $n-t$ supporters for $\sigma w$ and hence will be confirmed. Note that one supporter is $w$, the other is $v$ and the remaining are all the $n-t-2$ correct children of $\sigma w v$. Also note that $w$ is confirmed, so all $n-t$ correct will be confirmed on $(\sigma,w,d)$.\\

\item If $u \neq w$ has a set $|V|=n-t$, such that for each $v' \in V$, $u$ is a supporter of $v'$ on $(\sigma,w,d)$ and $v'$ is confirmed on $(\sigma,w,d)$ \textbf{then}
$u$ is a voter of $(\sigma,w,d)$.

\textit{Note:} this is somewhat similar to the notion of a Voter in grade-cast (\cite{FM97,FM88}). But there is a crucial difference: all the $n-t$ echoers need to be \textit{confirmed}. Also note that $w$ is a voter for itself.\\

\item \itrule: 
If $w$ has a set $|U|=n-t$, such that for each $u' \in U$, $u'$ is a voter of $(\sigma,w,d)$ \textbf{then} if $\sigma w \notin \RT$, then put $\RT(\sigma w):=d$ and color descendants of $\sigma w$ with $d$ as well.

\textit{Note:} this is somewhat similar to the notion of a grade 2 in grade-cast. A crucial difference is that the $n-t$ voters needed are defined with respect to \textit{supported} echoers. This is a non-trivial change that breaks the standard grade-cast properties. Also note that we not only put a value in $\sigma w$ but also color all the descendants. \\

\item \lastroundrule: 
if $\IT(\sigma w)=d$ and $\sigma \in \Sigma_{t}$ \textbf{then} if  $\sigma w \notin \RT$, then put $\RT(\sigma w):=d$.

\textit{Note:} this is a standard rule to deal with the last round.

\endsmall{enumerate}	

\textbf{Rules for $\RT$ tree resolve:}
The following definitions and rules cause a node to be resolved based only on information in $\RT$ (these rules do not look at $\IT$).

\beginsmall{enumerate}

\item If there is a set $|U|=t+1$, such that for each $u' \in U$, $\RT(\sigma w v u')=d$  \textbf{then} we say
$v$ is $\RT$-confirmed on $(\sigma,w,d)$.

\textit{Note:} if any correct sees a node as confirmed then it has $n-t$ that echo its value. At least $t+1$ of them are correct and they all cause all correct to see the node as $\RT$-confirmed. Of course a node may become $\RT$-confirmed even if it was never confirmed by any correct.
Observe that if $\RT(\sigma w u)=d$ then, by coloring, $u$ is $\RT$-confirmed on $(\sigma,w,d)$.\\

\item If $u \neq w$ has a set $|V|=n-t$, such that each $v' \in V$ is $\RT$-confirmed on $(\sigma,w,d)$ and
for each $v' \in V\setminus\{u\}$, $\RT(\sigma w v' u)=d$ 
and  if $u \in V$ then also $\RT(\sigma w u)=d$,
 \textbf{then}
$u$ is $\RT$-voter of $(\sigma,w,d)$.

\textit{Note:} if any correct process sees a node as a voter then it has $n-t$ echoers that are confirmed. So each of these $n-t$ echoers will be $\RT$-confirmed. So  all correct processes will see this node as $\RT$-voter. Of course a node can become $\RT$-voter even if it was never a voter at any correct process.\\

\item \gcrule: 
If $w$ has a set $|U|=t+1$, such that for each $u' \in U$, $u'$ is a $\RT$-voter of $(\sigma,w,d)$ \textbf{then}
if $\sigma w \notin \RT$, then put $\RT(\sigma w):=d$,   and color descendants of $\sigma w$ with $d$ as well. The rule applies also for node $\sigma w=\bar\epsilon$.

\textit{Note:} if any correct process does \itrule 
then this rule tries to guarantee that all correct processes will also put this node in $\RT$. 
The problem is that \srule 
(see below) may be applied to one of the echoers and this may cause some of the $\RT$-voters to lose their required support. The following rule fixes this situation. It reduces the threshold to $n-t-1$ but requires that all children nodes
are fixed.\\

\item \rgcrule: 
If all the children of $\sigma w$ are in $\RT$ (\ie  $\forall \sigma w v \in \Sigma$: $\sigma w v \in \RT$) and exists a set $|V|=n-t-1$, such that for each $v' \in V$, $\RT(\sigma w v')=d$, \textbf{then} if $\sigma w \notin \RT$, then put $\RT(\sigma w):=d$,   and color descendants of $\sigma w$ with $d$ as well. The rule applies only for nodes $|\sigma w|\ge 1$.

\textit{Note:} as mentioned above, the 
\rgcrule 
requires a threshold of $n-t-1$ so that it can take into account the possibility of one value changing to $\bot$ due to the following rule:\\

\item \srule: 
If there is a set $|V|=t+2 - |\sigma w u|$ such that for all $v \in V$, $\RT(\sigma w u v)=\bot$ and for all $u' \neq u$ such that $\sigma w u' \in\Sigma$, $\sigma w u' \in \RT$  \textbf{then} if $\sigma w u \notin \RT$, then put $\RT(\sigma w u):=\bot$,   and color descendants of $\sigma w u$ with $\bot$ as well. The rule applies only for  $|\sigma w u|\ge 2$.

\textit{Note:}  
This rule can be applied to at most one child.\\

\item \srootrule: 
If exists a set $|U|=t+1$ such that for each $u \in U$, $\RT(u)=\bot$ then if $\bar\epsilon \notin \RT$, then put $\RT(\bar\epsilon):=\bot$,   and color descendants of $\bar\epsilon$ with $\bot$ as well.

\textit{Note:} this rule is important in order to stop quickly if $t+1$ correct processes start with the value $\bot$.

\endsmall{enumerate}

To prevent the data structures from expanding too much processes close branches of the tree, and from that point on they do not send messages related to the closed branches. 
We use the notation $\{\sigma \in \RT[r]\}$ to denote an indicator variable that equals  true if $\RT(\sigma)$ was assigned some value by the end of round $r$, and false otherwise. 

\textbf{Branch Closing and Early Resolve rules:}
There are three rules to close a branch in  $\IT$ two of them also trigger an early resolve. By the end of round $r$, $r\le \phi,$  
\beginsmall{enumerate}
\item \decayrule: 
if $\mbox{$\exists$} \sigma' \sqsubseteq \sigma$ such that $\sigma' \in \RT[r-1]$, then close the branch  $\sigma\in\IT$.

\textit{Note:} this is the simple case: if a process already fixed the value of $\sigma'$ in $\RT$ in round $r-1$ then it stops in the end of round $r$, since by 
the end of
round $r+1$ all correct processes will put $\sigma'$ in $\RT$ (and will interpret this process's silence in the right way 
during
round $r+1$).  There is no need to continue. Coloring will fix all the values of this subtree.\\

\item  \earlyitrule: 
if $\sigma\in\Sigma_{r-1}$ and exists $U\subseteq N$, $U\cap\{u' \mid u'\in \sigma\}=\emptyset$,
$|U|=n-r$, 
such that for every $u,v \in U\setminus\F$,
$\IT(\sigma u)= \IT(\sigma v)$, then  if $\sigma \notin \RT$, then put $\RT(\sigma):=\IT(\sigma)$ and close the branch  $\sigma\in\IT$.

\textit{Note:} this is a case where the process can forecast that all correct processes will put $\sigma$ in $\RT$ in the next round (because the process sees that 
all children nodes 
agree). So the process can fix $\sigma$ in this round and stop now, because all correct processes will fix $\sigma$ in $\RT$ next round (and will interpret this process's silence in the right way).  \\

\item  \strongitrule: 
if $\sigma\in\Sigma_{r-2}$ and exists $U \subseteq N$, 
$U\cap\{u' \mid u'\in \sigma\}=\emptyset$,
$|U|=n-r+1$ such that for every $u,v \in U\setminus\F$, 
where $v\not=u,$
$\IT(\sigma uv)= \IT(\sigma vu)$ 
then,  if $\sigma \notin \RT$, then put $\RT(\sigma):=\IT(\sigma)$ and close the branch  $\sigma\in\IT$. 

\textit{Note:} in this case all the correct children of $\sigma$ except for at most one will be fixed in the next round to the same value, so the 
\rgcrule 
will be applied to $\sigma$ in the next round. So we can fix $\sigma$ in this round and stop now.\\

\endsmall{enumerate}

In each round all the above rules are applied repeatedly until none holds any more.

The rules above imply that there are two ways to give a value to a node in $\RT$. One is assigning it a value using the various rules, and the other is coloring it as a result of assigning a value to one of its predecessors.  We will use the term {\em color} for the second one and the term {\em put} for the first one.

\textbf{Rules for fault detection and masking:}
The following definitions and rules are used to detect faulty processes, put them into $\F$ and hence mask them (all messages from $\F$ are masked to $\bot$). The last rule also defines an additional masking.
The process first updates its $\F$ and $\FA$ sets using the sets received from the other processes during the current round. 
A  process is added to $\F$ or $\FA$  once it appears in $t+1$ or $2t+1$ sets, respectively.
Next the process applies the following fault detection rules.
The fault detection is executed before applying any of the resolve rules above.  When a new process is added to $\F$, the new masking is applied and the fault detection is repeated until no new process can be added.  Only then the resolve rules above are applied. 

At process $z$ by the end of round $r$:

\beginsmall{enumerate}

\item \textbf{Not Voter}: If $\exists \sigma w \in \Sigma_{r-1}$ and $w \neq z$ and $\not \exists \sigma' \sqsubseteq \sigma w$ such that $\sigma' \in \RT$ and it is not the case that there exists a set $|U|=n-t-1$ such that for each $u' \in U$, $\IT(\sigma w u') = \IT(\sigma w)$ \textbf{then} add $w$ to $\F$.

\textit{Note:} this is the standard detection rule after one round - if anything looks suspicions then detect. \\

\item \textbf{Not IT-to-RT}: If $\exists \sigma w \in \Sigma_{r-2}$ for which  
$w$ does not have a set $|U|=n-t$, such that for each $u' \in U$, $u'$ is a voter of $(\sigma,w,d)$, and
$\not \exists \sigma' \sqsubseteq \sigma $ such that $\sigma' \in \RT$  \textbf{then} add $w$ to $\F$.

\textit{Note:} this is the standard detection rule after two rounds - if anything looks suspicions then detect.\\

\item If $u,$ $u \neq w$, has a set $|V|=n-t$, such that for each $v' \in V$, $u$ is a supporter of $v'$ on $(\sigma,w,d)$  \textbf{then} we say that 
$u$ is an \textit{unconfirmed voter} of $(\sigma,w,d)$.

\textit{Note:} the notion of an \textit{unconfirmed voter} is exactly that of a voter in the standard grade-cast protocol.\\

\item  If $w$ has a set $|U|=t+1$, such that for each $u' \in U$, $u'$ is an unconfirmed voter of $(\sigma,w,d)$ \textbf{then} we say that $\sigma w$ is \textit{leaning towards} $d$.

\textit{Note:} the notion of \textit{leaning towards} is exactly that of getting grade $\ge 1$ in  the standard grade-cast protocol.\\

\item \textbf{Not Masking}: If $\sigma w \in \Sigma_{r-3}$ is leaning towards $d$ and there exists $u$, $|V|=t+1$, and $d' \neq d$  such that for each $v' \in V$, $\IT(\sigma w u v')=d'$ and there exists $|\sigma''|>|\sigma|$ such that $\IT(\sigma''w u) \neq \bot$ then
\beginsmall{enumerate}
\item
$\IT(\sigma''w u) =\bot$;
\item
if by the end of the round $\not \exists \sigma' \sqsubseteq \sigma'' w$ such that $\sigma' \in \RT$ \textbf{then} add $u$ to $\F$.
\endsmall{enumerate}

\textit{Note:}  If $\sigma w$ is leaning towards $d$ then $u$ must have heard at least $t+1$ say $d$ on $\sigma w$. If  $t+1$ say $u$ said $d'$ then $u$ must have said $d'$ to some correct. So $u$ must have received $d'$ from $\sigma w$ but in the next round $u$ hears $t+1$ say $\sigma w$ said $d$. So $u$ must conclude that $w$ is faulty and $u$ must mask him from the next round. If $u$ did not mask some $\sigma'' w u$ then the \textbf{Not Masking} rule will detect $u$ as faulty and mask all such $\sigma'' w u$ 
for you and also mark you as faulty. The reason we wait until the end of the round to add that node to $\F$ is that it might be a node of a correct process that stopped in the previous round and hence did not send any messages in the current round, and therefore did not send masking. In such a case we mask its virtual sending, but do not add it to $\F.$

\endsmall{enumerate}

\textbf{Finalized Output:}
By the end of each round (after applying all the resolve rules), 
the process checks whether there is a frontier  in $\RT$. A \textit{frontier} (also called a cut) is said to exist if for all $\sigma \in \Sigma_{\phi+1}$ there exists some sub-sequence $\sigma' \sqsubseteq \sigma$ such that $\sigma' \in \RT$. 
\beginsmall{enumerate}
\item  \textbf{Early Output rule}:
By the end of a round, if $\bar\epsilon\in\RT$, {\bf output} $\RT(\bar\epsilon)$. 

\item \textbf{Final Output rule}:
Otherwise, if there is a frontier, 
{\bf output} $\bot.$

\endsmall{enumerate}

Observe that the existence of a frontier can be tested from the current $\IT$ in $O(|\IT|)$ time.

\textbf{Stopping rule:} If all branches of $IT$ are closed, stop the protocol.

\section{The Consensus Protocol Analysis}\label{sec:protocol}

The EIG protocol implicitly presented in the previous section is a consensus protocol $\Df$, where
$\phi $, 
$1 \le \phi \le t$ is a parameter. Protocol $\Df$ runs for at most $\phi+1$ rounds and solves Byzantine agreement against a $(t,\phi)$-adversary.
Denote by $G$ the set of correct processes, $|G|\ge n-t$, where $n=|N|$, and by $S$, $S=\bigcap_{q\in G}\F_q$,  the set of processes that are masked to $\bot$ by all correct processes.
Let $s:=|S|$.  

Our solution invokes several copies of the EIG protocol.  For each invoked protocol, $\Df$, there are two cases: either 
 $s\ge t-\phi$, or we are guaranteed that the input of all correct processes that start the protocol is the same
(in particular, it may be that some correct processes have halted and do not start the protocol).  The following lemma deals with this latter case.

\begin{lemma}\label{lem:input-agree}[Validity and Fast Termination]
For any $(t,t)$-adversary, and $n\ge 3t+1$,
\beginsmall{enumerate}
\item
if every correct process that starts the protocol holds the same input value $d$ then $d$ is the output value of all correct processes that start the protocol, by the end of round 2, and all of them complete the protocol by the end of round 3.
\item
if all correct processes start the protocol and $t+1$ correct processes start with $\bot$ then all correct processes output $\bot$ by the end of round 3 and stop the protocol by the end of round 4.
\item For $p,q\in G$, no $p$ will add $q$ to $\F_p$ in either of the above cases.
\endsmall{enumerate}
\end{lemma}

\ifdefined\LONG

\begin{proof}
To prove the first item, 
let us follow the protocol.  
Let $G_1$ be the set of correct processes that start the protocol and let $G_2=G\setminus G_1$, be the remaining correct processes that remain silent throughout the protocol.

Initially, for every $z\in G_1$, $\IT_z(\bar\epsilon)=d_z.$

In the 1st round every correct process $z\in G_1$ sends $\langle \bar\epsilon, z, d_z \rangle$ to every process.
By the end of the 1st round, every correct process applies the receive rule for  all the other  processes. 
Thus, every correct process $z\in G_1$ has $\IT_z(x) := d_x$, for every $x\in G,$ since it completes the missing values from correct processes in $G_2$ to be its own input value.
Thus, the receive rule assigns at each $z\in G_1$, $\IT_z(\sigma x):=\IT_z(\sigma)$ for a missing value by $x\in G_2$ for $\sigma$.
\earlyitrule,  may be applied by some correct processes at the end of the first round, and as a result will put $\RT(\bar\epsilon)=d$ and will output $d$.

Since $\IT_z(x)=d$ for all $x\in G,$
by the end of the 1st round, every $z\in G_1$ sees every $x\in G$ as supporter of $x$ for $(\bar\epsilon,\bar\epsilon,d)$. 

In the 2nd round, every correct process $z\in G_1$ that did not apply \earlyitrule by the end of the 1st round, sends $\langle x, z, d \rangle$ for every process $x\in G$ to
every process.  Again, if any correct process did not send a message, its  missing value for any $x\in G$ will be assigned the same value at all correct processes. Notice that some additional correct processes may not send in the second round.

By the end of the 2nd round, after applying the receive rule,  at each $z\in G_1$ that did not apply \earlyitrule by the end of the 1st round,  $\IT_z( xy)=d$ for every $x,y\in G.$ 
Thus, 
for every such $x$, every $y\in G\setminus\{x\}$, is a supporter of $x$ for $(\bar\epsilon,\bar\epsilon,d)$. 
As a result, 
for the set $|G|=n-t$, for each $u' \in G$, $u'$ is a supporter of $v$ for $(\bar\epsilon,\bar\epsilon,d)$, for every $v\in G$. 
Therefore, every $v\in G$ is confirmed on $(\bar\epsilon,\epsilon,d)$.
Therefore, 
every process $z\in G_1$, that did not apply \earlyitrule by the end of the 1st round,  sees every process $x\in G$ as a voter of $(\bar\epsilon,\bar\epsilon,d)$.  
This implies that it can apply the \itrule and will put $\RT(\bar\epsilon)=d,$  will output $d$, and will stop the protocol by the end of round $3.$

For the second claim: by the end of the 1st round,  every correct process $z$ has $\IT_z(x) := \bot$, for at least $t+1$ processes $x\in G.$ 
Let $A=\{x\mid x\in G\ \&\ d_x=\bot\}.$  If $\bot$ was the input value to all correct processes, we are done by the previous claim. Otherwise, no correct process will apply \earlyitrule to a value that is not $\bot.$

In the 2nd round, every correct process $z$ sends $\langle x, z, d_x \rangle$ for every process $x\in G$
 to every process.  By the end of the 2nd round, 
 after applying the receive rule,  at each $z\in G$, $\IT_z( xy)=d_x$, for every $x,y\in G.$
Thus, every process $z\in G$ sees each process  $v\in A$  both as supporter of $v$ for $(\bar\epsilon,v,\bot)$, 
and also as supporter of $u$ for $(\bar\epsilon,v,\bot)$
for every $u\in G$.

In the 3rd round, every correct process $z$ sends $\langle vx, z, \bot \rangle$ for every process $v\in A$ and $x\in G$
 to every process.
 If any correct process applied \earlyitrule in the previous round, then its missing value regarding other correct processes will be identical at all correct processes.
 By the end of the 3rd round, 
 after applying the receive rule,  at each $z\in G$, $\IT_z( vxy)=\bot$, for every $v\in A$, and $x,y\in G.$
 Thus,
 every process $z\in G$ sees every process $u\in G\setminus \{v\}$ as a supporter of $u'$ for $(\bar\epsilon,v,\bot)$, for every $u'\in G\setminus\{v\}$ and $v\in A$. For every such $v$ and $u'$, $v$ is also a supporter of $u'$ for $(\bar\epsilon,v,\bot)$. Thus, every such $u'$ is confirmed on $(\bar\epsilon,v,\bot)$, for every $v\in A$.  Moreover, by definition, every such $v$ is also confirmed on $(\bar\epsilon,v,\bot)$.
 
As a result, every process $z\in G$ sees every process $u\in G$ as a voter to $(\bar\epsilon,v,\bot)$, for every $v\in A$. Thus, $v\in\RT_z$ for every $v\in A$.
 Thus, it can apply the \srootrule and will put $\RT(\bar\epsilon)=\bot$ by the end of round $3$, and will stop by the end of round $4$.
 
To prove the 3rd claim, observe that the fault detection rules can be applied only in rounds $2$ or $3$. If a correct process did not send any message in round $2$, it is because of applying  \earlyitrule, and it's missing values will not cause any other correct process to be suspected as a faulty process, neither the correct process that did not send.
By the end of the 2nd round $\bar{\epsilon}$ will be in $\RT_q$ for every $q\in G$, and no one will apply any fault detection rules anymore.
 
If all correct processes  participated in round $2$, then Not-Voter will not apply to any correct process.  If any correct process did not send any message in round $3$, it's missing values will not harm any correct process or itself and all correct processes  will be in $\RT$ by the end of the round.  For similar reasons, the Not-Masking rule will not cause any correct process to be added to $\F.$
\end{proof}

\else 
\fi

The only case in which not all correct processes invoke  a $\Df$ protocol 
is when some of the background running monitors are being invoked by some of the correct processes, while others may have already stopped. This special case is guaranteed to be when the inputs of all participating correct processes is $\bot$, and consensus can be still be achieved. \lemmaref{lem:input-agree} implies the following:

\begin{corollary}\label{cor:partial}
For any $(t,t)$-adversary, and $n\ge 3t+1$, 
if every correct process that invokes the protocol start with input $\bot$, then $\bot$ is the output value at each participating correct process by the end of round 2, and each participating correct process completes the protocol by the end of round 3. Moreover, for $p,q\in G$, no $p$ will add $q$ to $\F_p$.
\end{corollary}

The gossip exchange among correct processes about identified faults ensures the following:
\begin{lemma}\label{lem:FA}
For a $(t,\phi)$-adversary and protocol $\Df$, $n\ge 3t+1$, assuming \propertyref{prop:init-fail}, for any $k$, $1\le k\le \phi+1$, 
by the end of round $k$, for every two correct processes $p,q$, $\FA_p\subseteq\F_q$ and $\FA_p[k-1]\subseteq\FA_p[k]$. 
\end{lemma}

\ifdefined\LONG
\begin{proof}
	Prior to invocation the claim holds by  \propertyref{prop:init-fail}. In each round processes exchange their $\F$ sets.  If a process finds out that some process $b$ appears in the lists of at least $t+1$ processes it adds $b$ to $\F$, and if it appears in $2t+1$ lists it adds it to both $\F$ and $\FA.$ 
	The $\F$ and $\FA$ sets are never decreased, and $\FA$ is updated only through gossiping. Therefore, it is easy to see that by the end of each  round the claim holds. 
\end{proof}
\else 
\fi

A node may initially assign a value using one of the ``put'' rules and later it may color it to a different value.  In the arguments below we sometimes need to refer to the value that was put to a node rather than the value it might be colored to.  
Once a node has a value it is not assigned a value using any put rule any more. 
Thus, the value assigned using a put rule is an initial value that may be assigned to a node before it is colored, or  that node may never have a value put to it.  
To focus on these put operations, we will add, for proof purposes, that whenever a node $p$ uses a put rule for some $\sigma$, except \lastroundrule, it also puts $\sigma$ in  $\PT_p$ (The ``Put-Tree'') and as a result at that moment, $\PT_p(\sigma)=\RT_p(\sigma)$.  
We do not color nodes in $\PT_p$, thus for $\sigma$ that is colored, but was not assigned a value prior to that,  $\PT_p(\sigma)$ is undefined. We exclude \lastroundrule from $\PT$ on purpose.


The following is the core statement of the technical properties of the protocol. The only way we found to prove all these is via an induction argument that proves all properties together.  The theorem contains four items. 

The detection part proves that correct processes are never suspected as faulty.  The challenge is that the various rules instruct processes when to stop sending messages, and that might cause other correct processes to be suspected as faulty.

The validity part proves that if a correct process sends a value, it will reach the $\RT$ of every other correct process within two rounds. It also proves that if a correct process decides not to send a value (thus, closed a branch), the appropriate node will be in $\RT$ of every correct process. 
The third claim in the validity part is that if a process appears in $\FA$, then it appears in $\RT$ of every correct process within two rounds. 

The safety part intends to prove consistency in the $\RT$.  The challenge is that coloring may cause the trees of correct processes to defer.  
Therefore the careful statements looks at $\PT$, and which rule was used in order to assign the value to it. 
The $\bot$ value is a default value, therefore there is a special consideration of whether the value the process puts is $\bot$ or not.  
The end result is that if a node appears in $\PT$ of two correct processes, it carries the same value.

The liveness part shows that if a node appears in $\RT$ of a correct process, it will appear in $\RT$ of any other correct process within two rounds.

\begin{theorem}\label{thm:main}
For a $(t,\phi)$-adversary and protocol $\Df$, $n\ge 3t+1$, assuming \propertyref{prop:init-fail} and that all correct processes participate in the protocol, then for any $1\le k\le \phi+1:$

\beginsmall{enumerate}
\item\label{c:good-safe} {\bf No False Detection:} For $p,q\in G$, no $q$ will add $p$ to $\F_q$ in round $k$.

\item\label{c:good-sending} {\bf Validity:} 
\beginsmall{enumerate}
\item\label{c:good-sent}  For $\sigma\in\Sigma_{k-3}$ if $p \in G$,  sends $\langle \sigma, p, d_p \rangle$, then at the end of round $k$, at  every correct  process $x$, either $\RT_x(\sigma p)=d_p$ or $\exists \sigma' \sqsubset \sigma$ such that $\sigma' \in \RT_x$.  For $k=\phi+1$, the property holds also for any $\sigma\in\Sigma_{k-2}$ and for any $\sigma\in\Sigma_{k-1}$.

\item\label{c:masking}  If $z \in \FA$ in the beginning of round $k-2$, then by the end of round $k$, at  every correct process, either $\RT(\sigma z)=\bot$ or $\exists \sigma' \sqsubset \sigma$ such that $\sigma' \in \RT$.  For $k=\phi+1$, the property holds for $z \in \FA$ in the beginning of rounds $k-1$ or $k$.

\item\label{c:good-silent} For $\sigma\in\Sigma_{k-1}$, if $p \in G$, does not send $\langle \sigma, p, d \rangle$ for any $d\in D$, then at the end of round $k$,  at  every correct process $x$,  
 $\exists \sigma' \sqsubset \sigma $ such that $\sigma' \in \RT_x$.
\endsmall{enumerate}

\item\label{c:safety} {\bf Safety:} For $p,q\in G$, $x\in N$, $|\sigma x|\le\phi,$$\sigma x\in\PT_p[k]$, then
\beginsmall{enumerate}

\item\label{c:lock2d} if $p$ applies \gcrule to put  $\PT_p(\sigma x)=d,$ $d\neq \bot $, and $v$ is one of the  $\RT$-confirmed nodes on $(\sigma, x, d)$ in $\RT_p$ used in applying this rule in $\RT_p$, and in addition $\PT_q(\sigma x v)=\bot$, then $q$ applied \srule to put $\sigma x v$;

\item\label{c:not-bot} if  $|\sigma x|\ge1$ and  $\PT_p(\sigma x)=d,$ $d\neq \bot $, then,  by the end of round $k$,   $|V_q|\le t,$ where $V_q=\{u\mid \PT_q(\sigma xu)=\bot\}$;

\item\label{c:no-srule} if  $|\sigma x|\ge1$ and  $\PT_p(\sigma x)=\bot $ and it wasn't put using \srule, then, by the end of round $k$,   $|V_q|\le t,$ where $V_q=\{u\mid \PT_q(\sigma xu)\neq\bot\}$;

\item\label{c:full-safety} 
if  $\sigma x\in\PT_q[k],$ then $\PT_p(\sigma x)=\PT_q(\sigma x)$.
\endsmall{enumerate}

\item\label{c:two-rounds}{\bf Liveness:} For $p,q\in G$,  if $\sigma \in\RT_p[k-2]$ then $\sigma \in\RT_q[k]$.  For $k=\phi+1$,  if $\sigma \in\RT_p$ then $\sigma \in\RT_q$.

\endsmall{enumerate}

\end{theorem}

\ifdefined\LONG
\begin{proof} [Proof of \theoremref{thm:main}]
	We  prove the theorem by induction on $k$. 
	We first prove the theorem assuming $\phi>1$ and will conclude by proving the theorem for the case $\phi=1.$
	
	As the proof is quite complex, we split it into three ranges, $k=1$,  $k \le \phi-1$, and $k \le \phi+1$. We will prove the following claims, where each handles the appropriate range:
	
	\begin{claim}\label{c:range1}
\theoremref{thm:main} holds for $k=1$.
	\end{claim}
	
	The general case. This is where most of the technical challenge lies: 
	
	\begin{claim}\label{c:range2}
\theoremref{thm:main} holds for $1< k \le \phi-1$.
	\end{claim}
	
	The final two rounds, when the resolve rules are slightly different:
	
	\begin{claim}\label{c:range3}
\theoremref{thm:main} holds for $\phi-1< k \le \phi+1$.
	\end{claim}

	\begin{proof}[Proof of \claimref{c:range1}]
We will prove each of the four items separately.

\begin{proof}[Proof of Item~\ref{c:good-safe} for \claimref{c:range1}](Detection)
By the end of round 1, a process may add another to $\F$ only through gossiping.  \propertyref{prop:init-fail} implies that no correct process will suspect any other correct process.  The rest of the fault detection rules are not applicable in the first round. 
\end{proof}

\begin{proof}[Proof of Item~\ref{c:good-sending} for \claimref{c:range1}](Validity)
For $k=1$,   
Statement~\ref{c:good-silent} and Statement~\ref{c:masking} vacuously hold, since there was no such round. 
For proving Statement~\ref{c:good-sent} observe that in the first round only \earlyitrule is applicable. Assume that a correct process $p \in G$ applies  \earlyitrule by the end of round $1$, thus
node $\sigma$ is $\bar\epsilon$, since $\sigma=\epsilon$.
This implies that for every $x\in N\setminus\F_p$, $\IT_p(x)=d$.  Since $G\cap\F_p=\emptyset,$ we conclude that for every correct process $q$, $d_q=d$, and by the end of the 2nd round, by \lemmaref{lem:input-agree}, all correct processes  will have $\RT(\bar\epsilon)=d=d_p$ and we are done.	
\end{proof}

\begin{proof}[Proof of Item~\ref{c:safety} for \claimref{c:range1}](Safety)
Only Statement ~\ref{c:full-safety} is applicable for $k=1$.  Notice that the only case in which $\sigma x\in\PT_p[1]$ is when $p$ applies \earlyitrule  in the end of the 1st round and  as a result puts some value $d$ to the root node in its $\RT_p$. In such a case, it is clear that if $\sigma\in\PT_q[1]$ then $\PT_p(\sigma)=\PT_q(\sigma)$.
\end{proof}

\begin{proof}[Proof of Item~\ref{c:two-rounds} for \claimref{c:range1}](Liveness)
This item vacuously holds.
\end{proof}

This completes the proof of  \claimref{c:range1}.
	\end{proof}
	
	Now we move to proving the main part of the theorem.
	\begin{proof}[Proof of \claimref{c:range2}]
	The proof is by induction. The base case is  \claimref{c:range1}.
Assume correctness for any $k''$, $1\le k''<k,$  and we will prove the claim for $k$, $k\le \phi-1.$ 

\begin{proof}[Proof of Item~\ref{c:good-safe} for \claimref{c:range2}](Detection)
The fault detection takes place in every round before any resolve rule is applied.  
By induction we know that a  correct process will not add another correct process to $\F$ using gossiping from other processes.   The three rules to add a process to $\F$ are based on the messages accumulated in $\IT$. 
The induction on $k-1$ allows us to determine what messages correct processes will be sending in round $k$. 

Let round $k$ be the first round at which a process $p$ is not sending messages related to the branch of $\sigma.$  
There are three cases in which a correct process, $p$, stops sending, by using \decayrule, \earlyitrule and \strongitrule. If $p$ closes the branch of $\sigma$ at the end of  round $k-1$ and is not sending messages related to it in round $k$, the receive rule instructs correct processes what values to add to their $\IT$.

Let's consider the three fault detection rules. Not-Voter is not applicable, since in the previous round $p$ sent its messages appropriately.  Since $p$ is correct every correct process that sends messages echo's the message it sent, and whenever a correct process applies the receiving rule to assign messages to processes that did not send messages in the current round it adds the message $p$ originally sent. For the similar reason Not-IT-to-RT is not applicable.  

The last fault detection rule is Not-Masking.  Assume that a correct process $q$ is expecting process $p$ to mask away some process $w.$  
The Not-Masking rule allows $q$ to mask the non-sending by $\bot$, but $q$ will not add $p$ to  $\F$ if by the end of the round $q$ will have $ \exists \sigma' \sqsubseteq \sigma'' w$ such that $\sigma' \in \RT$.  Thus, $p$ will not be in $\F$ during the processing of all the rules below. 
Statement~\ref{c:good-silent} that is proved next guarantees that also by the end of the round a correct process $p$ will not be added to $\F.$
\end{proof}

\begin{proof}[Proof of Item~\ref{c:good-sending} for \claimref{c:range2}](Validity)

For Statement~\ref{c:good-sent}, the case of $k=\phi+1$ is excluded for now. Assume that $p$ sends $\langle \sigma, p, d_p \rangle$ in round $k-2$.  
If any correct process $x$ is not sending a message $\langle \sigma, x, d_x \rangle$, then by the protocol it should have set either
$\sigma'\in\PT_x[k-4]$ (if used \decayrule)  or $\sigma'\in\PT_x[k-3]$ (if used \earlyitrule or \strongitrule)  for some  $\sigma' \sqsubset \sigma$, and we are done by induction (Statement~\ref{c:full-safety}).  
If there is a correct node $x\in\sigma$, then the claim holds by induction (Statement~\ref{c:good-sent}). So we are left with the case  that no correct node appears in $\sigma$ and all correct processes  are participating in round $k-2$.  
By the end of round $k-2$ every correct process $x$ will apply the receiving rule and will have $\IT_x(\sigma p)=d_p.$

If any correct, $x$ ($x\not= p$), doesn't send a message $\langle \sigma p, x, d_p \rangle$ then we are done by induction, using similar argument as above.   Therefore, by the end of round $k-1$, every correct process will apply the receiving rule and will have $n-t-1$ children nodes for $p$ in its $\IT$.  
Thus,  by the end of round $k-1$, for every $x\in G$, at every $y\in G$, $x$ is a supporter of $x$ for $(\sigma, p, d_p).$ And for every $x,y\in G$, where $x\not= y \not= v$, $\IT_x(\sigma p y)=d_p$.  
In round $k$ some correct process (including $p$)  may not send messages and all the rest will send identical value $d_p$ messages.  The above implies that the receiving rule will  assign
to each correct process that does not send messages the identical value $d$ at every correct process that still process messages for this branch. 

As we argued before, since $p$ itself is confirmed on each node it echoes, every correct process will be a voter and therefore, 
by the end of round $k$, at every correct process $x\in G$, that still process messages for this branch either $\RT_x(\sigma p)=d_p$, or $\exists \sigma' \sqsubset \sigma p$ such that $\sigma' \in \RT_x$.

The proof of Statement~\ref{c:masking} is identical to the above, as if it is the case of a correct process sending $\bot$.

Proving Statement~\ref{c:good-silent}:
Let $\sigma\in\Sigma_{k-1}$ and $p\in G$. If $p$ does not send any message $\langle \sigma, p, d \rangle$ for any $d\in D$ in round $k$, then either the branch was closed earlier and we are done by induction, or this is the first round any correct process doesn't send a message on this branch. 
Thus, $p$ applied \decayrule, \earlyitrule or \strongitrule by the end of round $k-1.$ 
 
We will cover each of the closing rules separately. 

Proving the claim in case $p\in G$ uses \decayrule:   
by definition $\mbox{$\exists$} \sigma' \sqsubset \sigma$ such that $\sigma' \in \RT_p[k-2]$, which results in closing the branch by the end of round $k-1$ and not sending in round $k$. 
If $\mbox{$\exists$} \sigma'' \sqsubset \sigma$, $\sigma''\in\RT_p[k-3]$, then we are done by induction.  
Otherwise, it must be because of messages received in round $k-2.$  All such messages are reflected in $\IT_p$.  
To influence a $\sigma'\in\RT_p,$ it should be as a result of applying \itrule, \lastroundrule, \earlyitrule, or \strongitrule. 
Since $k-2\neq\phi+1,$ we conclude that it is not a result of applying \lastroundrule.  
If it is a result of $p$ applying  \earlyitrule, or \strongitrule in round $k-2$ then this branch would be closed already be the end of round $k-1$ and we are done by induction. 
Similarly, if any other correct process closed the branch by the end of round $k-2$, we are done by induction.

Assume now the case that it is a result of $p$'s using \itrule. 
Thus, there should be some $\bar\sigma w,$ such that $\sigma'\subseteq\bar\sigma,$ $\bar\sigma w\in\Sigma_{k-4}$ and $p$ applied \itrule in round $k-2$ to put it in $\RT_p$ (and $\PT_p$, for proof purposes). 
Let $d$ be the value assigned by $p$ to $\PT_p(\bar\sigma w)$ as a result of processing $\IT_p$ by the end of round $k-2.$.	
If there is a correct node in $\bar\sigma w$, we are done by induction.  
Since this is not the case, then when $p$ applied  \itrule it observed a set $U$ of $n-t$ processes in $\IT_p$ that are voters  of $(\bar\sigma, w, d)$, of which at least $t+1$ are correct processes.  	
Let $\bar U$ be the set of correct voters in $U$.

For each voter $v\in \bar U$ there is a set of $W_v$ of $n-t$ processes that are confirmed on $(\bar\sigma, w, d)$, where $v$ is a supporter to each $u\in W_v$ on $(\bar\sigma, w, d)$. Since we assume that there is no correct nodes in $\bar\sigma w,$ $v\not=w$.

By definition, 
for each  $u\in W_v\setminus\{w,v\}$, $\IT_p(\bar\sigma w u v)=d$, and since $v\in G$ and no correct process closed the branch or stopped sending yet, then
by the end of round $k-2$,  for every $x\in G\setminus\{u,v\}$, $\IT_x(\bar\sigma w u v)=d$. If $v\in W_v$, then all will also have $\IT_x(\bar\sigma w v)=d$. 

For $u\in G,$ since $\bar\sigma w\in\Sigma_{k-4}$, by induction, $\RT_x(\bar\sigma w u v)=d$, or $\exists \sigma' \sqsubset \bar\sigma w u v$ such that $\sigma' \in \RT_x$, at every $x\in G.$

For $u\not\in G,$ by the end of round $k-1$, for every $x\in G$, at every $y\in G$, $x$ is a supporter of $x$ for $(\bar\sigma w u, v, d).$ And for every $x,y\in G$, where $x\not= y \not= v$, $\IT_x(\bar\sigma w u v y)=d$.  
In round $k$ some correct process (including $p$)  may not send messages and all the rest will send identical value $d$ messages.  The above implies that the receiving rule will  assign
to each correct process that does not send messages the identical value $d$ at every correct process that still process messages for this branch. 

As we argued before, since $v$ itself is confirmed on each node it echoes, every correct node will be a voters and therefore, 
by the end of round $k$, at every correct process $x\in G$, that still process messages for this branch, and for every $u\in W_v$, either $\RT_x(\bar\sigma w u v)=d$, or $\exists \sigma' \sqsubset \bar\sigma w u v$ such that $\sigma' \in \RT_x$.

Now observe that each  $u\in W_v$, being confirmed on $(\bar\sigma, w, d)$, has a set $U_u$ of $n-t$ of  supporters in $\IT_p$ of $u$ for $(\bar\sigma, w, d)$ (one of which is $u$ itself). 
Let $\bar U_u$ be the set of correct processes in $U_u$.  
By definition, for each  $u\in W_v$, $\IT_p(\bar\sigma w u)=d$ 
and for each $u'\in\bar U_u\setminus\{u\}$, $\IT_p(\bar\sigma w u u')=d$.  
Since no correct process closed the brach or stopped sending,  at every $x\in G$,
$\IT_x(\bar\sigma w u u')=d$, and if $u\in G$, then $\IT_x(\bar\sigma w u)=d$.  
Thus,  by the end of round $k$, 
at every correct process $x\in G\setminus\{u, u'\}$, that still process messages for this branch, $\PT_x(\bar\sigma w u u')=d$, where $u\in W_v$, and $u'\in\bar U_u\setminus\{u\}$.
Thus, each  $u\in W_v$, is $\RT$-confirmed on $(\bar\sigma, w, d)$ and each $v\in \bar U$ is $\RT$-voter on $(\bar\sigma, w, d)$.
The same holds, by definition, for $u$ and $u'$ if they did not closed the branch earlier.
This implies that such $x$ will apply \gcrule to assign $\PT_x(\bar\sigma w)=d$ (or would observe by that time $\exists \sigma' \sqsubset \bar\sigma w$ such that $\sigma' \in \RT_x$), which completes the proof for this case.

Proving the claim in case $p\in G$ uses \earlyitrule:
Assume that a correct process $p \in G$ applies \earlyitrule  by the end of round $k-1$.
Let $\sigma\in\Sigma_{k-2}$ and denote $\sigma=\tau u$.
The assumption of $p$'s closing the branch implies, among other things,  that for every $x,y\in N\setminus\F_p$, such that $\tau ux , \tau uy \in\Sigma_k,$  $\IT_p(\tau u x)=\IT_p(\tau u y)=d$, for some $d\in D,$ and thus $\PT_p(\tau u)=d$.
This also implies that every correct process $x$ that applies the receiving rule in round $k$ will assign $\IT_x(\sigma p)=\IT_x(\sigma)=d$. If  there is any correct process in $\tau$, we are done by induction (Statement~\ref{c:good-sent} on $k-1$, since the correct processes sent in $k-3$ or earlier).
If this is not the case, whether $u$ is correct or not,  we conclude that by the end of round $k-1$ every correct process $x\in G$ will have $\IT_x(\tau u y)=d$ for every $y\in G\setminus\{u,x\},$ and  if $u\in G$ then also $\IT_x(\tau u)=d$.
This is true since by \lemmaref{lem:FA}, and Item~\ref{c:good-safe}, $G\cap\F_q=\emptyset.$
Thus,  by the end of round $k$ every correct process $x$ that did not close the branch will use \itrule to obtain $\RT_x(\tau u)=d$ (or would observe by that time $\exists \sigma' \sqsubset \tau u$ such that $\sigma' \in \RT_x$), and we are done.

Proving the claim in case $p\in G$ uses \strongitrule: 
Assume that $p$ applies \strongitrule  by the end of round $k-1$. 
If there is a correct process in $\sigma$, we are done by induction. 
If $\mbox{$\exists$} \sigma' \sqsubseteq \sigma$ such that $\sigma'\in\RT_q[k-1]$ for any correct $q$, we are also done.
Otherwise, let $\sigma\in\Sigma_{k-3}$. 
By definition there exists $U$, $U\cap\sigma=\emptyset$, $|U|=n-r+2$ such that for every $u,v \in U\setminus\F$, 
where $v\not=u,$
$\IT_p(\sigma uv)= \IT_p(\sigma vu)$.  
Let $x$ be the node such that $\sigma x\in\Sigma$, but $x\not\in U\cup\F.$ 
Since we assume that there is no correct process in $\sigma$, $G\subseteq U\cup\{x\}.$  
Assume first that $x$ is not correct.  
If this is the case, then the assumption on $U$ implies that all members of $U$ are supporters and voters and by the end of round $k-1$, $\sigma$ would be in $\RT$ of every correct process.  
If this is not the case, we are left with the option that $x$ is correct but doesn't agree with  some of the values all members of $U$ sent. 
Denote by $\bar U$ the correct member of $U$, and it is clear that $|\bar U|=n-t-1$ and $|U|\ge n-t.$ 
The definition of the set $U$ implies that by the end of round $k-1$, either $p$ puts $\sigma$ in $\PT_p$, or $\exists \sigma' \sqsubset \sigma$ such that $\sigma' \in \RT_p$.
Moreover, by induction,  for every member $u$ of $\bar U$, $\sigma u\in \RT_q$ of every correct process $q$ by the end of round $k.$  
Thus, by the end of round $k$ every correct process $q$ that doesn't already have $\sigma\in\RT_q$ will be able to apply \rgcrule to put $\bar\sigma\in\RT_q$, and we are done.
\end{proof}

\begin{proof}[Proof of Item~\ref{c:safety} for \claimref{c:range2}](Safety)
Notice that when a process $p$ puts a value to a node $\sigma x$, say in round $k$, then at that point in time $\mbox{$\not\exists$} \sigma' \sqsubset \sigma \text{, such that } \sigma' \in \RT_p[k] $.  

Observe that if both $p$ and $q$  put values to $\sigma x$ prior to round $k$, then the claims hold by induction on $k$.
Therefore we limit ourselves to  nodes which value $q$ puts in its $\PT$ in round $k$ and $p$ had put a value to that node in its $\PT$ in some round $k' \le k.$ Moreover, we limit ourselves to the case where no correct process had put a value to that node in its $\PT$ in any round $k''<k'.$

We prove Item~\ref{c:safety} by backward induction on the length $\ell=|\sigma x|$ from $\ell=k$ to $1$. 
For each $\ell$ we will go through all the put rules $p$ could have applied in setting the value to $\sigma x$ in round $k$ or earlier, and for each rule we consider the relevant rules  $q$ could have apply, and we will prove that the four statements hold in each case.

The rules to put a value to a node in $\RT$ (and $\PT$) are:
1) \itrule,
2) \gcrule,
3) \rgcrule,
4) \srule,
5) \srootrule,
6) \earlyitrule,
7) \strongitrule
and
8) \lastroundrule.

The case $\ell=k$:  A node of level $k$, where $k\le\phi $, cannot be put in $\PT\,$ by the end of round $k$. 

The case $1\le\ell<k$:  let $|\sigma x|=\ell$ and assume correctness for every $\ell'>\ell$ . Since $k<\phi$, \lastroundrule is not applicable.

If there is a correct predecessor in $\sigma$, we are done by Item~\ref{c:good-sending}, since by the end of round $\ell+1$ process $q$ will have $\sigma x\in\RT_q$ (due to coloring), hence no node $\sigma x u$ will be in $\PT_q$ and all four statements clearly hold.

Otherwise, if process $x$ is correct, by  Item~\ref{c:good-sending}, by the end of $\ell+2$ process $q$ will have $\sigma x\in\RT_q$. A node $\sigma x u$ can be in $\PT_q$ only if $q$ applied \earlyitrule,  or \strongitrule  in that round, so any such node will also be set to the value of $\sigma x$, which is the same at both $p$ and $q,$ thus all four statements hold.

Otherwise, there is no correct process in $\sigma x$. 
Thus, node $x$ has $n-\ell$ children nodes, out of which at least $n-t$ are correct and out of the $t-\ell$ others, at most $\phi-\ell$ are actively faulty and at lease $t-\phi$ are silent.

We start by proving the first three statements and after that we will prove the fourth statement.

 \ding{228}  Consider the case that $p$ used \gcrule to put $\sigma x$:
 \gcrule implies that there are $t+1$ $\RT$-voters.  Each $\RT$-voter has a set of $n-t$ children nodes $\RT$-confirmed for $(\sigma, x, d^\prime)$.  
 Define, in such a case, by $V_p$ the set of children nodes of $\sigma x$ that are $\RT$-confirmed to $d^\prime$ in $\PT_p$. By definition,   each confirmed node in $V_p$ has $t+1$ children nodes in $RT_p$ with the same value $d^\prime\!.$

\begin{proof} [Proof of Statement~\ref{c:lock2d} of Item~\ref{c:safety} for \claimref{c:range2}]
 By definition, confirmed is defined for $\ell+2\le k < \phi+1$.
 For node $v$ being $\RT$-confirmed implies that there is a set $V_d$, such that for each $v'\in V_d$, 
 $\RT_p(\sigma x v v')=d,$ where $|V_d|=t+1.$ 
If $\sigma x v\in \RT_p$ when $p$ puts $\sigma x$, then also $\sigma x v\in \PT_p$, otherwise $\sigma x$ should be in $\RT_p$ already.  
Moreover, if $\sigma x v\in \RT_p$, it should be that $\RT_p(\sigma x v)=d$, otherwise, by coloring, $\RT_p(\sigma x v v')$ would also not be equal $d$. 
By induction, level $\ell+1$, Statement~\ref{c:full-safety}, we conclude that $q$ can't put $\sigma x v$ to $\bot.$  Therefore, it should be the case that when $p$ puts $\sigma x$ to $\PT_p$, $\sigma x v \not\in\PT_p.$  In such a case, all children nodes of $\sigma x v $ that are in $\RT_p$ are in $\PT_p.$ Specifically, every $v'\in V_d$ is in $\PT_p.$ This also implies that $v\not\in G.$
 By induction, on level $\ell+2$, Statement~\ref{c:full-safety}, we conclude that for every $v'\in V_d$, if $\sigma x v v' \in\RT_q$ then $\PT_p(\sigma x v v')=\PT_q(\sigma x v v')=d.$ 
 
Node $v$ has exactly $n-\ell-1$ children nodes. When $q$ puts a value to node $\sigma x v$, all children nodes of node $\sigma x v$ are not colored. 
 There are at most $n-\ell-1-(t+1)<n-t$ children nodes of $\sigma x v$ in $\PT_q$ that  are not in $V_d.$
 
 Look at the rules $q$ may use in order to put $\sigma x v $ to $\bot.$ 
 
 --- Consider the case that $q$ used \itrule to put $\sigma x v$ to $\bot$: 
 If $q$ applies \itrule, then it should have for each voter $e$ a set $U_e$ of $n-t$   processes confirmed on $(\sigma x, v,\bot)$ in $\IT_q$. There is at least one process in the intersection of $U_e$ and $V_d$.  Denote it by $u$, $u\in U_e\cap V_d.$  
 Observe that the definition of $V_d$ implies that $u\not=v$.
 Being confirmed implies that $u$ has a set of at least $t+1$ correct processes $U_u$ such that $\IT_q(\sigma x v u u')=\bot$ for every $u'\in U_u.$
 For any such $u'$ that sends messages in this round, $\IT_p(\sigma x v u u')=\bot$.  
 If there is $u'$ that closed the branch using \decayrule, then it did so before round $k-2$, and we are done by induction.
 Otherwise it used   \earlyitrule and both $q$ and $p$ would assign it the same value, and therefore we also conclude that $\IT_p(\sigma x v u u')=\bot$.
 If it used \strongitrule, then there is a set $U'$ of at least $t+1$ correct processes such that $\IT_p(\sigma x v u u')=\IT_q(\sigma x v u u')=\bot$, since all but one send the same value, and $q$ saw $n-t$ of them.
 
 We now argue that if $p$ has such a set of children node, it implies that if $\sigma x v u \in \PT_p$, then $\PT_p(\sigma x v u)=\bot.$ 
 
Consider the various put rules $p$ can use to put a value to  $\PT_p(\sigma x v u)$. 
Thus, if $p$ uses  \earlyitrule in round $\ell+3$ it should be to the value $\IT_p(\sigma x v u)=\bot.$ 
 If $p$ applies \itrule in round $\ell+4$ it should be the case that $\IT_p(\sigma x v u)=\bot.$ 
 By the end of round $\ell+5$, all the correct children nodes in $U_u$ (or $U'$), by  Item~\ref{c:good-sending}, will be in $\PT_p$ with value $\bot$ and will color their subtrees in $\RT_p$ to $\bot$.  Therefore, if $p$ applies any rule to put the value of $\sigma x v u$, it will be to $\bot.$ 
This contradicts the fact that $u\in V_d.$

 ----  Consider the case that $q$ used \gcrule to put $\sigma x v$ to $\bot$: 
 If $q$ applies \gcrule, then it should have a set $U_e$ of $\RT$-confirmed on $(\sigma x, v,\bot)$ in $\PT_q$.
 Each $u$ in $U_e$ has a set $W_u$ of size $t+1$ such that for each $u'\in W_u$ $PT_q(\sigma xvuu')=\bot.$ 
 Since, at least one of the nodes in $U_e$ is in $V_d$, there is a contradiction to the induction on Statement~\ref{c:not-bot}.

 ---  Consider the case that $q$ used \rgcrule to put $\sigma x v$ to $\bot$: 
 Contradiction, to Statement~\ref{c:full-safety}.
 
 Thus, we are left with the option of $q$ applying \srule  put $\sigma x v$ to $\bot,$
 proving the statement.
\end{proof}

\begin{proof}[Proof of Statement~\ref{c:not-bot} of Item~\ref{c:safety} for \claimref{c:range2}]
 In this case, potentially some nodes from $V_p$ (though at most one) may resolve to $\bot.$  Observe that node $\sigma x$ in $\RT_q$ has at most $n-\ell-(n-t)=t-\ell$ children nodes outside $V_p.$  Since $\ell\ge 1,$ for the claim not to hold there should be at least 2 nodes form $V_p$ that resolve to $\bot$. Statement~\ref{c:lock2d} and the definition of \srule imply that at most one node can be resolved using \srule. We are done since $t-\ell+1<t+1.$
\end{proof}

\begin{proof}[Proof of Statement~\ref{c:no-srule} of Item~\ref{c:safety} for \claimref{c:range2}]
 The observation above implies that every node in $V_p$ that is put in $\RT_q$ should be with value $\bot$. Thus, proving this case.
\end{proof}

\ding{228} Consider the case that $p$ used \itrule to put $\sigma x$:\\
Statement~\ref{c:lock2d} is not applicable in this case, and the rest of the cases we discuss next.

\begin{proof}[Proof of Statement~\ref{c:not-bot} of Item~\ref{c:safety} for \claimref{c:range2}]  
 The \itrule implies that in $\IT_p$ there is a set $V$ of $n-t$ voters of $(\sigma, x, d)$, where $d\neq\bot$. Each $v\in V$ has a set $W_v$ of $n-t$ processes that are confirmed on $(\sigma, x, d)$, where $v$ is a supporter to each $u\in W_v$ on $(\sigma, x, d)$. Each  such $u$, being confirmed on $(\sigma, x, d)$, has a set $U_u$ of $n-t$ supporters in $\IT_p$ to $u$ on $(\sigma, x, d)$. 
 Each of these sets of size $n-t$ contains at least $t+1$ correct nodes. Let $U_{\!\tiny{\mbox{$\bot$}}}$ be the set of children nodes, where each one has at most $t$ correct supporters to  $(\sigma, x, d)$  in $\IT_p.$ The above implies that $|U_{\!\tiny{\mbox{$\bot$}}}|\le t-\ell$ (notice that  $U_{\!\tiny{\mbox{$\bot$}}}\subseteq N\setminus W_v$).
 
 Assume by contradiction that $q$ has $|V_q|> t$ children nodes of $\sigma x$ in $\PT_q$ such that $\PT_q(\sigma x u)=\bot$  for each $u\in V_q$.  Therefore, there must 
 exist  two  nodes, $y_1,y_2\in V_q$  that are not from the set $U_{\!\tiny{\mbox{$\bot$}}}$ (because $\ell\ge1$ so $|U_{\!\tiny{\mbox{$\bot$}}}|\le t-\ell\le t-1$).

We now go through the put rules $q$ can apply to put values to the children nodes of $\sigma x$. We will also study the minimal round at which $q$ can apply these put rules.
Since  $\RT_q(\sigma x)$ can't have a value when the put rule is applied by $q$ to the children nodes of  $\sigma x$, the earliest round at which $q$ can use any other rule to put it's value, if at all, is the end of $\ell+2$ .
 
For round $\ell+2:$ 
If $\ell+2\le k$, then by the end of round $\ell+2$, it can't be that all echoing processes to $y_1$ or $y_2$ have sent the  value $\bot$.
 Thus, by the end of round $\ell+2$, $q$ cannot have  either $y_1$ or $y_2$ in $\RT_q$ with value $\bot$, and the claim holds.

For round $\ell+3:$
If $\ell+3\le k$, then by the end of round $\ell+3$, each $y\in\{y_1,y_2\}$ has $t+1$ correct children that are supporters for $d$ in $\IT_q$ and therefore $y$ can't be in $\PT_q(\sigma x y)$ for value $\bot$.
Thus, by the end of round $\ell+3$, $q$ can have at most $t-\ell$ children nodes of $\sigma x$ in $\RT_q$ with value $\bot$, and the claim holds.
 
 If $\ell+4\le k$, then by the end of round $\ell+4$, by  Item~\ref{c:good-sending}, the value of all correct
 nodes in all sets  $V$, $W_v$ and $U_u$ above are already in $\RT_q$. This implies that $y_1$ and $y_2$ each has
 at least $t+1$ children nodes in $\RT_q$ with value $d$. The value of neither $y_1$ nor $y_2$ can be put to $\bot$
 using rules \gcrule, or \rgcrule, and clearly not \srootrule. We already excluded \earlyitrule, \strongitrule, and \itrule, so the
 only rule that may be applied is \srule. But \srule can be applied only when all other sibling nodes are already in
 $\RT_q,$ so it can be applied to either $y_1$ or $y_2$ but not to both. A contradiction. This completes the proof of
 Statement~\ref{c:not-bot} for this case, assuming $p$ used \itrule to put the value of $\RT_p(\sigma x).$
\end{proof}

\begin{proof}[Proof of Statement~\ref{c:no-srule} of Item~\ref{c:safety} for \claimref{c:range2}]
 The proof is identical to the proof of Statement~\ref{c:not-bot} with a small change, except that \srule does not produce a value that is different than $\bot$. 
 \end{proof}

\ding{228} Consider the case that $p$ used  \earlyitrule  in order to put $\sigma x$.

\begin{proof}[Proof of Statement~\ref{c:not-bot} of Item~\ref{c:safety} for \claimref{c:range2}]  
The \earlyitrule implies that in $\IT_p$ there is a set 
 $U$, $U\cap\{u' \mid u'\in \sigma x\}=\emptyset$,
$|U|=n-\ell$, 
such that for every $u,v \in U\setminus\F$, 
$\IT(\sigma x u)= \IT(\sigma x v)$.
Assume first that no correct process closes the branch by the end of round $\ell+1.$
This implies that $q$ will also see all correct processes sending the same value.  Therefore, it can't apply any rule on it's $\IT_q$ to put any child of 
$\sigma x$ $\PT_q$ with a value of $\bot$.   
By the end of round $\ell+3$, by Item~\ref{c:good-sending}, the value of all correct
 nodes in $U$ are 
 already in $\RT_q$. This implies that $q$ can't have a set $V_q$ of more than size $t$ for any different value.
 
Now, if there is $u\in U$ that closed the branch and did not send in round $\ell+1$, then by the end of round $\ell+1$, 
by Statement~\ref{c:good-silent},
 at every correct process $q$, $\exists \sigma' \sqsubset \sigma x $ such that $\sigma' \in \RT_q$, which implies that by the end of $\ell+1$, 
$ \sigma x \in \RT_q$, contradicting our assumption.
\end{proof}

\begin{proof}[Proof of Statement~\ref{c:no-srule} of Item~\ref{c:safety} for \claimref{c:range2}]
 The proof is identical to the proof of Statement~\ref{c:not-bot} with a small change, except that \srule does not produce a value that is different than $\bot$. 
\end{proof}	
 
 \ding{228} Consider the case that $p$ used  \strongitrule in order to put $\sigma x$. 
 
\begin{proof}[Proof of Statement~\ref{c:not-bot} of Item~\ref{c:safety} for \claimref{c:range2}]  
Assume for contradiction that  $\exists V_q$ such that $|V_q|= t+1,$ where $V_q=\{u\mid \PT_q(\sigma xu)=\bot\}$.
The \strongitrule implies that in $\IT_p$, by the end of round $\ell+2$, there is a set 
$U$, 
$U\cap\{u' \mid u'\in \sigma\}=\emptyset$,
$|U|=n-\ell+1$ such that for every $u,v \in U\setminus\F$, 
where $v\not=u,$
$\IT_p(\sigma x uv)= \IT_p(\sigma x vu)=d$.

Assume first that no correct process closes the branch by the end of round $\ell+2.$
This implies that $\IT_q$ will include all values above appearing in $\IT_p$ for correct processes. 
We assume that there is no correct process in $\sigma x$, and that  $|\sigma x|\ge 1$.  Therefore $|\F_q\setminus\{u' \mid u'\in \sigma x\}|<t$. 
Moreover, also $|(\F_p\cup\F_q)\setminus\{u' \mid u'\in \sigma x\}|<t$. 
Therefore, there should be at least two processes in $V_q$ that are not in $\F_q$.
Therefore, there should be  $y\in V_q$ such that $y\in U\setminus(\F_p\cup\F_q)$.   
By the definition of $U$, there is a set $U_q$ of $n-t-1$ correct processes such that $\IT_q(\sigma x y u)=d$, for $u\in U_q$.  
Therefore, by the end of round $\ell+3$ $q$ can't have $\PT_q(\sigma x y)=\bot$.  

Since $p$ applied \strongitrule by the end of round $\ell+2$, by the end of $\ell+3$, by Statement~\ref{c:good-silent}, $\sigma x\in\RT_q$, so 
$\PT_q(\sigma x y)$ will never be set to $\bot$. A contradiction.
 \end{proof}

\begin{proof}[Proof of Statement~\ref{c:no-srule} of Item~\ref{c:safety} for \claimref{c:range2}]
 The proof is identical to the proof of Statement~\ref{c:not-bot} with a small change, except that \srule does not produce a value that is different than $\bot$. 
 \end{proof}

 \ding{228}  Consider the case that $p$ used \rgcrule to put $\sigma x$:  In this case when $p$ applies the rule, all its children nodes are in $\PT_p.$ By induction (Statement~\ref{c:full-safety}), none of these children nodes will appear with a conflicting value in $\PT_q.$ Since $n-t-1$ of them are with the same value $d^\prime$, then at most $n-\ell -(n-t-1)=t-\ell+1$ are with a different value. \rgcrule is applied only when $\ell\ge 1.$ This immediately implies that Statement~\ref{c:not-bot} and Statement~\ref{c:no-srule} hold. 

We completed the proof of the first 3 statements. We now prove the last one.

\begin{proof}[Proof of Statement~\ref{c:full-safety} of Item~\ref{c:safety} for \claimref{c:range2}]
 If $|\sigma x|\ge1,$ then Statement~\ref{c:not-bot} and Statement~\ref{c:no-srule} clearly prove that Statement~\ref{c:full-safety} holds, unless $p$ uses  \srule to put $\sigma x$. The proof above covers the case that $q$ uses any rule other than \srule, by symmetry between $p$ and $q$ in this statement.  Thus, we are left with the case that  both are using \srule, and clearly both put $\bot.$
 
 We are left to consider the case $|\sigma x|=0,$ thus $\sigma x = \bar\epsilon$. For that we need to consider all put rules that $p$ and $q$ may have applied.  
 There are 3 applicable rules, \itrule, \gcrule, and \srootrule,  to put a value to $\bar\epsilon$. Notice that \srule and  \rgcrule are not applicable and \earlyitrule,  or \strongitrule  were covered in Statement~\ref{c:good-silent}.

Node $\bar\epsilon$ has $n$ children nodes, out of which at least $n-t$ are correct and out of the $t$ others, at most $\phi$ are actively faulty and at lease $t-\phi$ are silent. Notice that for $\bar\epsilon$, every child node that is in $\RT_p$ is also in $\PT_p$, since once we assign a value to  $\bar\epsilon$ we do not process any other node.
  
 \ding{228} Consider the case that $p$ used \earlyitrule or \strongitrule to put a value to $\bar\epsilon$:
Both rules imply that $p$ sees a unanimous echoing by all $n$ processes, with the exception of at most one process.  Since we assume that all correct processes participate, there is no way that $q$ will put a different value to $\bar\epsilon$.
  
 \ding{228} Consider the case that $p$ used \itrule to put $\bar\epsilon$ and that $\RT_p(\bar\epsilon)=\bot$:
 The basic arguments are the same as in the case $\ell\ge1$, but the set of put rules that $q$ may apply differ. If $q$ also uses \itrule, then the claim clearly holds. If $q$ uses \srootrule, then it obtains the same value. So we are left with the case of $q$ using \gcrule.  The arguments are the same as in the case  $\ell\ge1$,  which exclude the possibility that $q$ puts any value other than $\bot$ to $\bar\epsilon$, completing the proof of this case.
 
 \ding{228} Consider the case that $p$ used \itrule to put $\bar\epsilon$ and that $\RT_p(\bar\epsilon)=d,$ $d\neq\bot$: We now need to consider the possibility of $q$ using \itrule, \gcrule and \srootrule.  The arguments for the  first two  are the same as above and are left out. 
 
For using \srootrule node $q$ should have a set $V_q$, $|V_q|=t+1,$ such that for each $v\in V_q,$ $\RT_q(v)=\bot.$  Notice that  also here there is no difference between $\RT_q(v)$ and $\PT_q(v).$

 The \itrule implies that in $\IT_p$ there is a set $V_p$ of $n-t$ voters of $(\bar\epsilon, \epsilon, d)$. 
 Each $v\in V_p$ has a set of $W_v$ of $n-t$ children nodes such that $v$ is a supporter to each $u\in W_v$ on $(\bar\epsilon, \epsilon, d)$, and each  such $u$ has a set $U_u$ of $n-t$ supporters in $\IT_p$ to $(\bar\epsilon, \epsilon, d)$. 
 Each of these sets of size $n-t$ contains at least $t+1$ correct nodes. Thus, there is a set $U_{\!\tiny{\mbox{$\bot$}}}$ of size at most  $t$ that does not have at least $t+1$ correct supporters to $(\bar\epsilon, \epsilon, d)$.
 
 Thus, there should be a process $x\in V_q$ that has  a set $U_x$ of $n-t$ supporters in $\IT_p$ to $(\bar\epsilon, \epsilon, d)$.  The set contains a set $\bar{U_x}$ of at least $t+1$ correct processes that are also supporters in $\IT_q$ to $(\bar\epsilon, \epsilon, d)$.  Consider the various rules $q$ can apply to put a value $\bot$ to $\RT_q(x).$ By the end of the 2nd round it can apply \earlyitrule  to set a $\bot$ to it, because all processes in  $\bar{U_x}$ send a different value.  For that reason it can't apply \itrule in the end of round 3 to put  value $\bot$ to $\RT_q(x)$. By the end of round 4 for every process $y\in\bar{U_x}$ $\PT_q(xy)=d.$ 
 Process $q$ can't apply \srule to put a value $\bot$ to $x$, since that rule is not applicable for $|x|=1.$ \gcrule, \rgcrule or \strongitrule, can't be used to put $\bot.$  
 Since we assume that $k<\phi$, the case $\phi=1$ is not relevant, Therefore also \lastroundrule can't be applied either - and we are done.
 
 \ding{228}  Consider the case that $p$ used \gcrule to put $\bar\epsilon$: If $q$ uses \itrule, by symmetry we are done.  If $q$ also uses \gcrule, by definition both obtain the same value.  We are left with the case that $q$ uses \srootrule. The interesting case is that $\RT_p(\bar\epsilon)=d$, $d\neq\bot.$  
 Observe that we cannot use the induction on $k=1$ since the set of applicable rules differ.
 The arguments for proving the case are similar to the previous case, the case of \itrule, since $q$ can't apply \srule to any node in level $1.$
 
 \ding{228}  Consider the case that $p$ used \srootrule to put $\bar\epsilon$: If $q$ also uses it the claim holds. Otherwise it falls into the other rules discussed above.
 
 This completes the proof of Statement~\ref{c:full-safety} .
\end{proof}

This completes the proof of Item~\ref{c:safety} (Safety) for \claimref{c:range2}.
\end{proof}

\begin{proof}[Proof of Item~\ref{c:two-rounds} for \claimref{c:range2} ](Liveness)
It is enough to prove that if $p\in G$ puts $\sigma x\in\PT_p$ in some round $r\le k$, then by the end of round $\max(r+2,\phi+1)$ $\sigma x\in\RT_q$, for every $q\in G.$

We prove the lemma by backward induction on $\ell=|\sigma x|,$ from $\ell=k$ to $\ell=1.$ 
As in the proof of  Item~\ref{c:safety}, the claim clearly holds for $\ell=k$, since  no node of level $k$, $k<\phi+1$ can be added to $\PT$ by the end of round $k.$  The case $\ell=k-1$ is applicable only to \earlyitrule, and is covered by the proof of Statement~\ref{c:good-silent}.

Assume the induction for any $k\ge \ell'>\ell$ and we will prove for $\ell,$ $\ell\le\phi-1.$  
If $\sigma x$ contains a correct node then by induction on Item~\ref{c:good-sending} we are done.  So assume that there is no correct process in $\sigma x.$
Let $p$ be the first to put $\sigma x$, where $\sigma x\in\PT_p$, and let $r$ be the round at which it did that.  Consider the various possible put rules.

\ding{228}  Case $p$ applied \earlyitrule,  or \strongitrule  Statement~\ref{c:good-silent} implies the proof.

\ding{228}  Case $p$ applied \itrule:  By definition, this can happen only in round $r=\ell+2$. The \itrule implies that there are $t+1$ correct
voters of $(\sigma,x,d)$ in $\IT_p$, each having $n-t$  nodes, each of which is confirmed on  $(\sigma,x,d)$ in $\IT_p$. Let $U_x$ be the set of the confirmed nodes on $(\sigma,x,d)$ in $\IT_p$ and $V_e$ the set of correct voters. Observe that $U_x$ contains at least $t+1$ correct processes.

If by round $r+1$ $\sigma x\in\RT_q$ we are done. If not, then if for any $u\in U_x$ $\sigma x u\in \RT_q$, it should be in $\PT_q$, and it should be with a value $d$, because of using either \itrule, \earlyitrule,  or \strongitrule  by $q$, and it can't obtain a different value, because of the correct processes in $U_x$ and $V_e.$

If by round $r+2$ $\sigma x\in\RT_q$ we are done. If not, 
Item~\ref{c:good-sending} implies that by $\max(r+2,k)$ all
voters in $V_e$ will appear in $\RT_q$ as $\RT$-voters on $(\sigma,x,d)$, since the nodes in $U_x$ will be confirmed to $(\sigma,x,d)$. These arguments and Item~\ref{c:safety} imply that if any of them is colored, it should be colored to $d$. Therefore,  $q$ can apply \gcrule  to add $\sigma x$ to $\PT_q$ and we are done.

\ding{228}  Case $p$ applied \gcrule: By definition, assuming that no branch closing took place, this can happen only in some round $r\ge\ell+4$.  Assume first that $\ell\ge 1$, we later deal with smaller values of $\ell$.
If by the end of round $r+2$ process $q$ puts a value to $\sigma x$ or to a predecessor of $\sigma x$, we are done.  Otherwise, by the induction hypothesis, by $r+2$, each node involved in applying \gcrule by $p$ to $\sigma x$ in $\PT_p$ is either colored or its value put by $q$ in $\RT_{q}.$ We will show that by $r+2$ process $q$ can apply one of the rules to put a value to $\sigma x$ in $\PT_q.$

Let  $V_p$ be that set of children nodes of $\sigma x$ that are $\RT$-confirmed to $d^\prime$ in $\RT_p$. By Statement~\ref{c:full-safety}, for every $v\in V_p$, if $\sigma xv\in\PT_p$ and $\sigma xv\in\PT_q$, then $\PT_p(\sigma xv)=\PT_q(\sigma xv)$.

None of the nodes in $V_p$ can be confirmed to a different  value than $d^\prime$ in $\RT_q$, unless it was put by $q$ to a different value.  If $d^\prime\neq\bot$ this can happen to the value of $\bot$ and by
Statement~\ref{c:lock2d}, this can happen only using \srule. Thus, there can be at most one such node $z\in V_p$ that was set to $\bot$ by $q$.

If none was set using \srule, then by $r+2$ process $q$ should see the same set of voters that $p$ did and is able to apply \gcrule. Otherwise, it should have applied the \srule to one of the nodes in $V_p.$
Before it can apply \srule, all other children nodes of $\sigma x$ should be put to a value. After applying \srule to node $\sigma x z$ all the children nodes of $\sigma x$ have a value in $\RT_q.$
For every $y\in V_p$ the value is $d^\prime$, so $q$, by that time, would have at least $n-t-1$ children nodes set to $d^\prime.$ Thus, $q$ can apply \rgcrule to put a value to $\sigma x$ and the claim holds.

In the case of $\ell=1$, by definition $q$ can't apply  \srule to set a value to $z$.  And therefore it should have been able to use \gcrule to set a value to $\sigma x$.
The case of $\ell=0$ is similar to the case of $\ell=1$.

\ding{228}  Case $p$ applied \rgcrule: since $p$ applies this rule, all the children nodes of $\sigma x$ are put to a value in $\PT_p$ and by induction by $r+2$ also at $q$. If  $\sigma x\in\RT_q$, we are done.  Otherwise, by Statement~\ref{c:full-safety} their value is the same as for $p$ and process $q$ can also apply \rgcrule.

\ding{228}  Case $p$ applied \srule or \srootrule: exactly as in the previous case.

This completes the proof of Item~\ref{c:two-rounds} (Liveness) for \claimref{c:range2}.
\end{proof}

This completes the proof of  \claimref{c:range2}.
	\end{proof}
	
	We can now complete the proof of the Theorem by covering the case of $k\in\{\phi,\phi+1\}$
	\begin{proof}[Proof of \claimref{c:range3}]
We cover both cases for each item.

\begin{proof}[Proof of Item~\ref{c:good-safe} for \claimref{c:range3}](Detection)
There is no special issues that surface in the last two round regarding detection, and the proof for the case $k<\phi$ holds.
\end{proof}

\begin{proof}[Proof of Item~\ref{c:good-sending} for \claimref{c:range3}](Liveness)

\ding{228} Consider the case $k= \phi$: There is no difference between the arguments for this case and those of $k<\phi.$ 

\ding{228} Consider the case $k= \phi+1$: If $|\sigma x|= \phi$ and if any correct process is not sending in this round it is because of applying the $\RT[r-3]$ limitation, and by induction we are done.  Otherwise, if $z$ sends, then \lastroundrule completes the proof. The case 
$|\sigma x|< \phi$ is identical to that of  $k< \phi$.

The proof of Statement~\ref{c:masking} is similar to the case in which a correct process sends $\bot$ in the first round (Statement~\ref{c:good-sent}).
\end{proof}

\begin{proof}[Proof of Item~\ref{c:safety} for \claimref{c:range3}](Safety)
Item~\ref{c:safety} is not applicable in case $k=\phi+1$.

Consider the case $k=\phi$. The case $|\sigma x|=\phi$: A value to a node at this level can't be put at any round $\le \phi.$
Observe that  two correct processes may put conflicting values in their $\RT$ to a node $\sigma xy$
at level $\phi+1$  that is associated with a faulty process, since they may have conflicting values in their
$\IT$ for that node. This may happen only if there wasn't any correct predecessor of $x$ in $\sigma$, since Item~\ref{c:good-sending} implies that before assigning $y$ a value it would already be colored.
By \propertyref{prop:init-fail}, there is no conflict on all the $t-\phi$ faulty nodes that are initially in $\FA.$ Thus, there can be at most one faulty node in level $\phi+1.$
Item~\ref{c:good-sending} also implies  that during round $\phi+1$  node $\sigma x$ will be assigned a value by all correct processes, and therefore so will node $\sigma xy.$

Statement~\ref{c:lock2d} is not applicable in the case of $|\sigma x|=\phi$.

\begin{proof}[Proof of Statement~\ref{c:no-srule} for \claimref{c:range3}]
 
 Node $\sigma x$ was put to $\bot$ by process $p.$ By the  assumption of Statement~\ref{c:no-srule}, \srule wasn't applied. \lastroundrule is not applicable, since we are in level $\phi.$  \itrule and \gcrule are not relevant, since there is only a single level of nodes in $\IT$ or $\RT$.
 \srootrule is relevant only for the case of  $\sigma x=\bar\epsilon$, which can't happen for $|\sigma x|=\phi$. 
 If  process $p$ uses \earlyitrule,  or \strongitrule, then similar arguments to those used in the proof of Statement~\ref{c:good-silent} can be used.
 
 We are left with  \rgcrule.  Node $\sigma x$ has $n-t-1$ children nodes in $\RT_p$ all having the value $\bot$ and all but one are clearly correct nodes.  Since there are exactly $n-\phi-1$ nodes in level $\phi+1$, and there can be at most $n-\phi-1-(n-t-2)=t-\phi+1$ nodes holding a non $\bot$ value.  Thus, node $q$ can't have $t+1$ or more children nodes with a value not $\bot$ when it applies it's put operation;  Completing the arguments for Statement~\ref{c:no-srule}.
\end{proof}

\begin{proof}[Proof of Statement~\ref{c:not-bot} for \claimref{c:range3}]
 
 Node $\sigma x$ was put to $d,$ $d\neq\bot$ by process $p.$ Thus, \srule is not applicable and, as in Statement~\ref{c:no-srule}, we are left with \rgcrule. Node $\sigma x$ has $n-t-1$ children nodes in $\RT_q$ all having the value $d$ and all but one are clearly correct nodes.  Since there are exactly $n-\phi-1$ nodes in level $\phi+1$, and there can be at most $n-\phi-1-(n-t-2)=t-\phi+1$ nodes holding a non $d$ value.  Thus, node $q$ can't have $t+1$ or more children nodes with a value not $d$ when it applies it's put operation;  completing the arguments for Statement~\ref{c:not-bot}.
\end{proof}

\begin{proof}[Proof of Statement~\ref{c:full-safety} for \claimref{c:range3}]
 
 Case $d=\bot$, if $\sigma x\in\PT_q$, then by Statement~\ref{c:no-srule}, the only applicable rules for $q$ are \rgcrule, or \srule. Both will result in validating the claim.  Case $d=\mbox{$\not\!\!\bot$}$, if $\sigma x\in\PT_q$, then by  Statement~\ref{c:not-bot} it is clear that the only possible rule to be applied is \rgcrule, which results in validating the claim for Statement~\ref{c:full-safety}.
\end{proof}

For node $|\sigma x|< \phi-1$ identical arguments to those used in the proof of \claimref{c:range2} complete the proof of Item~\ref{c:safety} (Safety) for \claimref{c:range3}.
\end{proof}

\begin{proof}[Proof of Item~\ref{c:two-rounds} for \claimref{c:range3} ](Liveness)
The arguments for this item are the same for $k=\phi$ and $k=\phi+1$.
As we mentioned before, it is enough to prove that if $p\in G$ puts $\sigma x\in\PT_p$ in some round $r$, then by the end of $\max(r+2,\phi+1)$ $\sigma x\in\RT_q$, for every $q\in G.$
The proof is by backward induction on $\ell=|\sigma x|.$  

Case $\ell=\phi+1.$ The only round at which a process can put a value to a node in level $\phi+1$ in it's $\RT$ is during round $\phi+1.$ At that round, every correct process that doesn't have $\sigma x$ in its $\RT$ as a colored node, will insert it to its $\RT$ using \lastroundrule.

Case $\ell=\phi$. Either $p$ used \earlyitrule,  or \strongitrule  or it set the value in  round $\phi+1.$
If it used \earlyitrule,  or \strongitrule, then Statement~\ref{c:good-silent} completes the proof. 
Now we need to consider the various potential put rules $p$ applied in round $\phi+1$ in order to put the value for $\sigma x.$  \itrule and \gcrule are  not applicable in this case.

\ding{228}  Case $p$ applied \rgcrule: If exists a correct process in $\sigma$ then by Item~\ref{c:good-sending} we are done. If
$x\not\in G$ then all children nodes of $\sigma x$ are either correct or silent, and if $p$ applies the rule, every correct process can apply the same rule.  If $x\in G$, then there are $n-t-1$ correct children nodes of $\sigma x$, all of which will send the same value, and all will apply the \rgcrule, completing the proof of this case.

\ding{228}  Case $p$ applied \srule: using similar arguments as above, this case is applicable only if $x\not\in G$, and as the arguments above show, if any correct process applies this rule, all will.

For node $|\sigma|< \phi-1$ identical arguments to those used in the proof of \claimref{c:range2} complete the proof of this case.
\end{proof}

This completes the proof of \claimref{c:range3}.
	\end{proof}
	
	We now prove the theorem for the case of $\phi=1.$
	
	By assumption there is at most one faulty process, say $b$, that doesn't appear in $\FA$ of any correct process. 
	There are at most  two rounds of information exchange.  
	
	In the first round every process sends its input value.  By the end of the first round, at every $z$, $\IT_z(\bar\epsilon)=d_z.$ $\IT_z(z)=d_z$, and for every $x\in N\setminus \F_z$, $\IT_z=d_x,$ where $d_x$ is the value received from $x$, and for every $y\in \F_z$, $\IT_z=\bot$.  
	The only rule that may be applied by a correct process by the end of this round is the \earlyitrule.  
	
	Assume that $z$ applies the \earlyitrule by the end of round 1. 
	This can happen only when all inputs are $\bot$ or when $F_z=\emptyset$ and all input values are identical.  If this happen, $z$ sets $\RT_z(\bar\epsilon)=\IT_z(\bar\epsilon)=d_z.$  $z$ does not send any message in round 2.  
	Following that, every correct process $p$ that doesn't stop sends to every process the set of values it entered to $\IT_p(x)$ for every $x\in N.$
	If a correct process $z$ stops, all these values are identical, other than the values associated with $b$.  
	Moreover, for $z$ and any other correct process that did not send a message, all correct processes add to their $\IT$ the same value for it. 
	By the end of round 2, every correct process, $p$, that did not stop applies \lastroundrule to copy $\IT_p(\sigma)$, for $\sigma\in\Sigma_2$ to $\RT_p(\sigma)$.  
	By the previous discussion it is clear that all will have identical values regarding all node, other than maybe the nodes on $\sigma$ that include $b.$  
	Therefore, every correct process $p$ will be able to apply \rgcrule and will put the same value to $\bar\epsilon.$  
	
	This discussion shows, implicitly that all the items of the theorem hold for this case.
	
	Now consider the case that no correct $z$ stopped at the end of round 1.  In the second round every correct process sends 
	to every process the set of values it entered to $\IT_p(x)$ for every $x\in N.$ By the end of round 2 every correct process, $p$,  applies \lastroundrule to copy $\IT_p(\sigma)$, for $\sigma\in\Sigma_2$ to $\RT_p(\sigma)$.   By the end of the second round, for every $p,q,z$ correct processes $\RT_p(zp)=\RT_p(zq)=\RT_q(zp).$ 
	
	Since $b$ is the only potentially faulty process we conclude that for every $p,q,z\in N\setminus \{b\}$, $\RT_p(bz)=\RT_q(bz).$ 
	
	We now show that the theorem holds in this case. 
	
	To prove 
	that Item~\ref{c:good-sending} (Validity) holds, let's look at its three statements. 
	Statement~\ref{c:good-silent}  vacuously holds. Statement~\ref{c:good-sent} holds, since for every correct process that sends in the last round there is consensus.  For every $p$ in $G$ that sends in the first round, as we mentioned before, all processes, but $b$, sent the same value $d_p$ that $p$ sent in the first round, and by applying \rgcrule, which can be applied to node $p$, all reach consensus.  For every $p\in\FA$, the same arguments hold.
	
	To prove that Item~\ref{c:safety} holds, let's look at its four statements. 
	
	\begin{proof}[Proof of Statement~\ref{c:lock2d} when $\phi=1$]  By definition, node $p$, $p\in G$, can apply \gcrule only on node $\bar\epsilon$. Assume it resolved to $d,$ $d\neq\bot.$  
By definition $p$ observed at least 2 processes as voters to $d$, and it identified $n-t$ $\RT$-confirmed nodes.  All correct processes among them will never resolve to $\bot.$  
The only possibility that another correct process $q$ can resolve any to $\bot$ is node $b$. If node $b$ is $\RT$-confirmed, it has at least 2 children nodes $x,y$ such that $\RT_p(bx)=\RT_p(by)=\bot$.  
Since both $x$ and $y$ are necessarily correct processes, we conclude that $\RT_q(bx)=\RT_q(by)=\bot$.  
Moreover, for all $\RT$-confirmed nodes $z$ in $\RT_p$, except of node $b$, $\RT_p(z)=\RT_q(z)=d,$ since all are correct.  
The only rule $q$ may be able to apply to resolve $b$ to $\bot$ is \srule.  But \srule is not applicable to nodes of level 1.
	\end{proof}
	
	\begin{proof}[Proof of Statement~\ref{c:not-bot} and Statement~\ref{c:no-srule}  when $\phi=1$] These statements clearly hold since $\PT_p$ is defined only for nodes in level 1, and $\PT_q$ is not defined for level $\phi+1$.
	\end{proof}
	
	\begin{proof}[Proof of Statement~\ref{c:full-safety}  when $\phi=1$] Consider three cases,  if $x\in G$, then by Item~\ref{c:good-sending} we conclude equality.  Consider the case that $x=b$.  In this case, as we wrote above, for every $p,q,z\in N\setminus \{b\}$, $\RT_p(bz)=\RT_q(bz).$  Therefore, if $p$ applied a rule to conclude $b\in \PT_p$, so will $q$.
We are left with the case of $x=\bar\epsilon.$ As we just proved, on every node of level 1, $p$ and $q$ agrees.  All but one of them are nodes associated with correct processes.  The only node on level 2 on which $p$ and $q$ differ is node $xb$.  But because of coloring, both color node $xb$ by the value of $x$.  Therefore, on every node $\sigma$, $|\sigma|\ge 1$ if $\sigma\in\RT_p$, then $\sigma\in\RT_q$.  Therefore, every rule $p$ applies holds also for $q$.  This completes the proof of Statement~\ref{c:full-safety}.
	\end{proof}
	
	To prove that Item~\ref{c:two-rounds} holds consider the 3 possible levels.  \lastroundrule implies that it holds for  level $\phi+1$. Statement~\ref{c:full-safety} proves the rest of the cases.
	
	To prove that Item~\ref{c:good-safe} holds observe that by \propertyref{prop:init-fail}, it holds initially.  In the first round, no detection takes place. In the 2nd round, no correct process suspects any other correct process.

	This completes the proof of \theoremref{thm:main}.
\end{proof}

\else 
\fi

The following Theorem summarizes the properties needed from our protocol. 

\begin{theorem}\label{thm:agree}
For a $(t,\phi)$-adversary and protocol $\Df$ and $n\ge 3t+1$ 
and assuming that all correct processes participate in the protocol:
\beginsmall{enumerate}
\item\label{t:agree}
Every correct process outputs the same value. 
\item\label{t:same}
If the input values of all correct processes  are the same, this is the output value. Every correct process outputs it by round $2$ and stops by round $3$.
\item\label{t:tp1}
If $t+1$ of the correct processes hold an input value of $\bot$, then all correct processes output $\bot$ by the end of round $3$ and stop by the end of round $4$.
\item\label{t:fp2}
If the actual number of faults is $f_\phi<\phi$, then all correct processes  complete the protocol by the end of round $f_\phi+2$. 
\item\label{t:f0}
If the actual number of faults is $f_\phi=0$, and all correct processes start with the same initial value, then all correct processes  complete the protocol by the end of round $1.$
\item\label{t:f1}
If the actual number of faults is $f_\phi=1$, and all correct processes start with the same initial value, then all correct processes  complete the protocol by the end of round $2.$
\item\label{t:within1}
If a correct process outputs in round $k$, it stops by the end of round $k+1$.
\item\label{t:within2}
If a correct process stops in the end of  round $k$, all correct processes output by round $k+1$ and stop by round $k+2$.
\endsmall{enumerate}
\end{theorem}

\ifdefined\LONG
\begin{proof} [Proof of \theoremref{thm:agree}]
	\ \\
	\noindent{\bf Proof of Statement~\ref{t:agree}}:
	By definition a correct process outputs a value once it identifies a frontier.  It is clear that by the end of round $\phi+1$ there is a frontier for every correct process.
	Define the front of $\RT$ to be: $\sigma x$ is in the front of $\RT$ if exists $p\in G$ such that $\sigma x\in\RT_p$ and for every $q\in G$,  $\sigma \not\in\RT_q$. \theoremref{thm:main} implies that if $\sigma$ is in the front of process $p$, within two rounds it will be in the front of any other correct process. Since all correct processes shares the front, then if $\bar\epsilon\in\RT_p$, it will be at every other correct process and vice versa.  Since a process does not stop for two rounds after it holds a frontier the first claim holds.
	
	\noindent{\bf Proof of Statement~\ref{t:same}}: \lemmaref{lem:input-agree} proves the second claim. 
	
	\noindent{\bf Proof of Statement~\ref{t:tp1}}:
	The proof of \theoremref{thm:main}  implies that \srootrule can be applied by the end of round 3  if there are $t+1$ correct processes that start with input $\bot.$ Thus, the third claim holds.
	
	\noindent{\bf Proof of Statement~\ref{t:fp2}}:
	Observe that if the actual number of faults is $f_\phi$ and $f_\phi<\phi,$ then for every $\sigma\in\Sigma_{\phi+1}$ there is a prefix of length $k$, $k\le f_\phi+1$ in which a correct process appears as the last node. If $k\le\phi-1$ then by \theoremref{thm:main}, by $k+2$ every correct process will have that prefix in its $\RT$ and will be able to apply \decayrule to close the branch by the end of round $\phi+2$.
	
	Consider a prefix $\tau p$ of length $\phi+1$.  
	By assumption $\tau$ contains all faulty  processes. 
	Therefore, by the end of round $\phi+2$, every correct process will be able to apply \earlyitrule to add $\tau p$ to $\RT$ and will close the branch.  Observe that sometimes more than one rule can be applied, but since we go down from the later rounds to the earlier ones, we happen to close the branch earlier.
	
	We are left with the case of $\tau p$ of length $\phi.$ There is at most one corrupt node, say $x$, that can send values relating to $\tau p$ that will be added to the $\IT$ of correct processes in rounds $\phi+1$ and round $\phi+2$. In round $\phi+1$ all correct processes  becomes children nodes and by the end of round $\phi+2$ all will add $\tau p$ to their $\PT$ and would be able to apply \strongitrule to close the branch.
	
	Thus, in all cases, by the end of round $\phi+2$ all correct processes  will close all branches and can output a value.
	
	\noindent{\bf Proof of Statement~\ref{t:f0}}: since there are no faults, all correct processes  apply \earlyitrule by the end of the first round to set a value to $\bar\epsilon$.
	
	\noindent{\bf Proof of Statement~\ref{t:f1}}: since there is a single fault, all correct processes  apply  \strongitrule by the end of the 2nd round  to set a value to $\bar\epsilon$.
	
	\noindent{\bf Proof of Statement~\ref{t:within1}}: The branch closing rules immediately imply that there can be at most one round between adding the final value to $\RT$ that produces the frontier, thus providing output, and closing of all branches that imply stopping the protocol.
	
	\noindent{\bf Proof of Statement~\ref{t:within2}}:  The first part of the statement holds, since if $p$ stops by the end of round $k$, it doesn't send anything in round $k+1$. \theoremref{c:good-sending} (Statement~\ref{c:good-silent}) imply that by the end of that round every correct process will output a value, and by the previous statement all will stop by the end of $k+2.$
\end{proof}
\else 
\fi


\section{Monitors}\label{sec:monitor}
We follow the approach of \cite{BG93-votes,GM93,GM98} with some modifications for guaranteeing early stopping.

In round $r=1$ we run $\mathcal{D}_t$ using the initial values. For each integer $k$, in round $1<r=1+4k < t-1$ we invoke  protocol $\mathcal{D}_{t-1-4k}$ whose initial values is either $\bot$ (meaning 
everything
is OK) or \bad (meaning that too many corrupt processes were detected). We call this sequence of protocols the {\em basic monitor sequence}.  
We will 
actually
run 4 such sequences.

\subsection{The Basic Monitor Protocol}
Each process $z$ stores two variables: $v \in D$, the current value, and $early$, a boolean value. Initially $v$ equals the initial input of process $z$ and $early :=\false$. Later, $early=\true$ will be an indicator that the next decision protocol must decide $\bot$ (because there is not enough support for \bad). Each process remembers the last value of $early_q$ it received from every other process $q$, even if $q$ did not send one recently.

Throughout this section we use the notation: $\bar{r}\equiv r \pmod 4$.
\begin{algorithm}[!ht]
\footnotesize
\SetNlSty{textbf}{}{:}

\lnl{m:1}  {\bf if} $\bar{r}= 1$:   \\

\nl\tb  \emph{if} $r<t-1$ \emph{then} \emph{invoke} protocol $\mathcal{D}_{t+1-r}$ with initial value $v_z$;\\

\lnl{m:2}  {\bf if} $\bar{r} = 2$: \\

\nl\tb at the end of the round:\\
\lnl{m:2a}\tb\tb  \emph{if} $|\FA| \ge r+3$ \emph{then} set $v_z:=\bad$

\nl\tb\tb\tb \emph{otherwise} set $v_z:=\bot$;\\

\lnl{m:3}  {\bf if} $\bar{r} = 3$:

\nl\tb send $v_z$ to all;\\

\nl \tb at the end of the round: \\
\nl\tb\tb \emph{if} $|\{q \mid v_q = \bad\}| \le  t$ \emph{then} set $early_z:=\true$

\nl \tb\tb\tb \emph{otherwise} set $early_z := \false$;\\

\lnl{m:0} {\bf if} $\bar{r}= 0$: 

\nl  \tb send $early_z$ 
to all;\\

\nl \tb at the end of the round: \\
\nl\tb\tb \emph{if} $|\{q \mid early_q=\true\}| \ge t+1$ \emph{then} 
set $v_z:=\bot$;

\nl\tb\tb \emph{if} every previously invoked protocol produced an output then set  $v_z:=\bot$.\\

\caption{
The Basic Monitor protocol (at process $z$)
}\label{alg:monitor}
\end{algorithm}

The monitor protocol runs in the background until the process halts.
The monitor protocol invokes a new $\D_\phi$ protocol every 4 rounds.
In each round, the monitor's lines of code are executed before running all the other protocols, and its end of round lines of code are executed 
before
ending  the current round in all currently running protocols. This is important, since it needs to detect, for example, whether all currently running protocols produced outputs for determining its variable for the next round. 
At the end of each round the monitor protocol applies the monitor\!\_halting and monitor\!\_decision rules below to determine whether to halt all the running protocols at once, or only to commit to the final decision value.

When a process is instructed to apply a monitor\!\_decision it applies the following definition. If it is instructed to halt (monitor\!\_halting), then if it did not previously apply the monitor\!\_decision, it applies monitor\!\_decision first and then halts all currently running protocols that were invoked by the monitor at once.

\begin{definition}[monitor\!\_decision]
A process that did not previously decide, {\bf decides} \bad, if any previously invoked protocol outputs \bad. Otherwise, it decides on the output of $\D_t$.
\end{definition}

When a process is instructed to  decide without halting, it may need to continue running all protocols for few more rounds to help others to  decide. We 
define
``halt by  $r+x$" to mean continue to run all active protocols until the end of round $\min\{r+x,t+1\}$, unless an halt is issued earlier.

\subsection{Monitor Halting and Decision Conditions}
Given that different processes may end various invocations of the protocols in different rounds we need a rule to make sure that all running protocols end by the end of round $f+2.$  
The challenge in stopping all protocols by the end of $f+2$ is the fact that individual protocols may end at round  $f+2$ and we do not have a room to exchange extra messages among the processes. 
This also implies that we need to have a halting rule at every round of the monitor protocol, since $f+2$ may occur at any round.

Each halting rule implies how  other rules need to be enforced in later rounds, since any process may be the first to apply a monitor\!\_halting  at a given round and we need to ensure that for every extension of the protocols, until everyone decides, all will reach the same decision despite the fact that those that have halted are not participating any more. The conditions take into account processes that may have halted. A process considers another one as halted if it doesn't receive any message from it in any of the concurrently running set of invoked protocols, monitors and the gossiping of $\F$.

To achieve that we add the following set of rules.

\noindent {\bf Monitor Halting Rules:}\\
\vspace{-1.5em}
\beginsmall{enumerate}
\item[$H_{\footnotesize\mbox{\sc \bad}}$.] Apply monitor\!\_halting if any monitor stops with output \bad.
Otherwise if any monitor outputs \bad, apply monitor\!\_decision now and monitor\!\_halting by $r+2.$

\item[$H_1$.\ \ \ \ ] Case $\bar{r}=1$:  
\beginsmall{enumerate}
\item
If all previously invoked protocols stopped, apply monitor\!\_halting. 
\item
Otherwise, if only the latest invoked protocol  did not stop and 
 $|\{q \mid early_q=\true  \mbox{ or $q$ halted}\}| \ge n-t$, 
 then apply monitor\!\_halting. 
 \item
Otherwise, if only the  latest invoked protocol  did not stop and 
$|\{q \mid early_q=\true  \mbox{ or $q$ halted}\}| \ge t+1$, 
 then apply monitor\!\_decision now and monitor\!\_halting by $r+2.$ 
\endsmall{enumerate}

\item[$H_2$.\ \ \ \ ] Case $\bar{r}=2$:  
\beginsmall{enumerate}
\item
 If all previously invoked protocols stopped, apply monitor\!\_halting.
 \item
Otherwise, if only the latest invoked protocol did not stop and 
 $|\{q \mid early_q=\true  \mbox{ or $q$ halted}\}| \ge n-t$
 was true in the previous round, 
 then apply monitor\!\_halting. 

 \item
Otherwise, if only the latest invoked protocol did not stop and 
$|\{q \mid early_q=\true  \mbox{ or $q$ halted} \}| \ge t+1$ 
was true in the previous round, then apply monitor\!\_decision and now and monitor\!\_halting by $r+1.$ 
\endsmall{enumerate} 

\item[$H_3$.\ \ \ \ ] Case $\bar{r}=3$: 
If all previously invoked protocols stopped, 
apply monitor\!\_halting.

\item[$H_4$.\ \ \ \ ]  Case $\bar{r}=0$: 
If all previously invoked protocols stopped and 
$|\{q\mid early_q=\true \mbox{ or $q$ halted}\}|\ge n-t$ 
then apply monitor\!\_halting.
\endsmall{enumerate}

\begin{lemma}\label{lem:haltf2}
If $n>3t$ and there are $f$, $f\le t$, corrupt processes then all correct processes apply monitor\!\_halting by the end of round $\min(t+1,f+2)$.
\end{lemma}

\ifdefined\LONG
\begin{proof}

We need to show that all previously invoked protocols halt by the end of round $\min(t+1,f+2)$. 
Observe that \theoremref{thm:agree} (Statement~\ref{t:fp2}), implies that $\D_t$ itself is stopped by $\min(t+1,f+2)$.

By definition, protocol $\D_\phi$ is invoked in round $r_\phi$, where $\phi=t+1-r_\phi.$ 
By \theoremref{thm:agree} (Statement~\ref{t:fp2}), $\D_\phi$ is stopped by $\min(\phi+1,t_\phi+2)$, if the upper bound on the number of faults (that were not detected by all correct processes before invoking the protocol) is $t_\phi$.  Note that if the number of faults that are not detected by all is higher than $t_\phi$ the protocol may not stop by $\phi+1$.  

Let's study the  number of faults that are not detected by all correct processes when $\D_\phi$ is invoked.  
\figureref{alg:monitor} Line~\ref{m:2} indicates that if any correct $p$ set $v_p:=\bad$ in round $r_\phi-3$, then, by \lemmaref{lem:FA}, the number of faults that are not detected by all correct processes when $\D_\phi$ is invoked is at most $t-r_\phi.$  
In such a case, by \theoremref{thm:agree},  $\D_\phi$ will be stopped by round $\min(\phi+1,t_\phi+2)$, where $t_\phi\le t-r_\phi.$ Let us call these $\D_\phi$ regular-protocols.

If no correct $p$ sets  $v_p:=\bad$, then all correct processes invoke $\D_\phi$ with $v=\bot$, therefore no matter how many faults are present (as long as not more than $t$), \lemmaref{lem:input-agree} guarantees that $\D_\phi$ is stopped within 3 rounds, and all outputs are obtained within 2 rounds. Let us call these $\D_\phi$ fast protocols.

For regular-protocols we need to prove that the extra conditions  hold.  In addition, for fast-protocols we need also to prove that the protocol that was invoked recently will also stop in time. 

Let us consider the $r\pmod 4$ round at which $\min(t+1,f+2)$ falls. 

\noindent {\bf Case $\min(t+1,f+2)\pmod 4=0$}: By $H_4$ we need to show that all previously invoked protocols will be stopped and that 
$|\{q\mid early_q=\true \mbox{ or $q$ halted}\}|\ge n-t$, at every correct process. 

For  regular-protocols, since all are stopped by round $\min(t+1,f+2)$ then when correct processes executed Line~\ref{m:2}, just before stopping, none would set $v:=\bad.$  Therefore, all will set $v$ to $\bot$ and later $early$ to $\true$.  Thus, the extra property for $H_4$ holds, and all will halt.

For fast-protocols, since no process sets $v$ to $\bad$, every previously invoked protocol stops within at most 3 rounds (\theoremref{thm:agree}, Statement~\ref{t:same}).  The latest protocol was invoked 3 rounds ago, and we are done. The arguments for the extra condition in $H_4$ are the same as for the regular-protocols.

\noindent {\bf Case $\min(t+1,f+2)\pmod 4=3$}: By $H_3$ we need to show that all previously invoked protocols will be stopped.

The arguments for regular-protocols 
and for fast protocols are the same, the latest invocation was two rounds ago, and therefore, by \theoremref{thm:agree} (Statement~\ref{t:same}), by the end of the current round all will be stopped.

\noindent {\bf Case $\min(t+1,f+2)\pmod 4=2$}: By $H_2$ we need to show that either all previously invoked protocols have stopped by the end of the current round, or all but the last one and the extra condition holds. 

If $\min(t+1,f+2)=t+1,$ then no protocol was invoked in the previous round, by definition.  All previous regular or fast protocols will be stopped by the end of the current round.  

If $\min(t+1,f+2)=f+2,$ by \theoremref{thm:agree} (Statement~\ref{t:fp2}), using similar arguments as above, all previous protocols will be stopped by the end of the current round, except, maybe the last protocol that was invoked in the previous round. Observe that correct processes set up their $v$ four rounds ago. Since the current round is $f+2$, then the round at which the processes executed Line~\ref{m:2} in \figureref{alg:monitor}  is $f-2$ and therefore no process could have more than $f$ faults, and would have set $v:=\bot.$ Therefore, every correct process  that haven't halt yet would send $early=\true$ two rounds ago, and therefore the extra condition for $H_2$ holds.

\noindent {\bf Case $\min(t+1,f+2)\pmod 4=1$}: By $H_1$ we need to show that either all previously invoked protocols have stopped by the end of the current round, or all but the last one and the extra condition holds. 

If $\min(t+1,f+2)=t+1,$ then no protocol was invoked in the current round, by definition.  All previous regular or fast protocols will be stopped by the end of the current round.

If $\min(t+1,f+2)=f+2,$ by \theoremref{thm:agree} (Statement~\ref{t:fp2}) using similar arguments as above, all previous protocols will be stopped by the end of the current round, except, maybe the last protocol that was invoked in the previous round. Observe that correct processes set up their $v$ three rounds ago. Since the current round is $f+2$, then the round at which the processes executed Line~\ref{m:2} in \figureref{alg:monitor}  is $f-1$ and therefore no process could have more than $f$ faults, and would have set $v:=\bot.$ Therefore, every correct process  that haven't halt yet would send $early=\true$ two rounds ago, and therefore the extra condition for $H_1$ holds.
\end{proof}
\else 
\fi

\begin{lemma}\label{lem:haltdecide}
If the first process applies monitor\!\_halting in round $r$ on $d$ then every correct process applies 
monitor\!\_decision by round $\min \{r+4, f+2 ,t+1\}$, applies monitor\!\_halting by round $\min \{r+5, f+2 ,t+1\}$, and obtains the same decision value,
$d$.
\end{lemma}

\ifdefined\LONG

\begin{proof}
Let $p$ be a correct process applying monitor\!\_halting in the earliest round that any correct process applies it. 

Observe that in some of the halting rules a process decides before the last invoked protocol outputs a value.  There may be cases that one process halts and other processes continue to run and even invoke an additional protocol after the halting.  We later prove that whenever these cases happen, the decision value is the same and it not \bad. We show that any protocol whose output is not taken into account by any correct process must output $\bot$.

Consider first the case that $p$ halts with output \bad.  
By \theoremref{thm:agree} (Statement~\ref{t:agree} and Statement~\ref{t:within2}), if  $p$  halts with output \bad
and if the output of that protocol is not ignored by any correct process then all correct processes will output \bad by next round and will halt within two rounds. This will lead to unanimous decision. 

So pending on the fact that we later prove that any protocol whose output is not taken into account by any correct process will output $\bot$, we are left to consider the case that $p$ does not output \bad.

If $r=\min(t+1,f+2)$, we are done by \lemmaref{lem:haltf2} (and  \theoremref{thm:agree}, Statement~\ref{t:agree}). Since every correct process considers the outputs of the same set of protocols,  the decision value is the same at every correct process.

Consider the various halting rules used by $p$ to apply monitor\!\_halting, and let $r$ be the round at which it was applied.

\noindent {\bf Case $p$ uses $H_1$}:  There are three possibilities, one in which $p$ noticed that all previously invoked protocols stopped.  In this case,   \theoremref{thm:agree} (Statement~\ref{t:within2}) implies that all correct processes 
will observe that all previously invoked protocols reported output by the end of $r+1$  and
will observe that all previously invoked protocols have stopped by the end of round $r+2$ and will use rule 
$H_3$ 
to apply monitor\!\_halting. All correct obtain the same decision value, since all will consider the same set of protocols and, by 
 \theoremref{thm:agree} (Statement~\ref{t:agree}) and the decision rule, will decide the same.

Otherwise, when $p$ executed round $r$ it noticed that by the end of that round all previous protocols stopped and only the one that started at the beginning of round $r$ did not stop yet
and the values of $early$ that $p$ received in round $r-1$ imply that   $|\{q \mid early_q=\true  \mbox{ or $q$ halted}\}| \ge n-t$.  Since no process halted earlier, in round $r-1$ every correct process  sets $v:=\bot.$ By \lemmaref{lem:input-agree}, the protocol that started in round $r$ will produce output of $\bot$ in round $r+1$ at all correct processes that did not stop earlier, and will stop by round $r+2.$ Thus, every correct process will apply either $H_2$ or $H_3$ and will reach the same decision.

Otherwise, when $p$ executed round $r-2$ it noticed that by the end of that round all previous protocols stopped and only the one that started at the beginning of round $r-2$ did not stop yet. 
Moreover, $p$ received at the beginning of round $r-3$,  $|\{q \mid early_q=\true  \mbox{ or $q$ halted}\}| \ge t+1$.  
Since no correct halted earlier, the instruction to set the value for $early$ implies that there was a correct process $q$ that set its $early_q$ to $\true$ in round $r-4$. 
Thus, $q$ received less than $t$ \bad. 
This implies that there are $t+1$ correct processes with $v=\bot$.  \lemmaref{lem:input-agree}, implies that the last protocol starting in the beginning of round $r-2$ will output value $\bot$ by the end of round $r$ and stop by the end of round $r+1$.  
By the end of round $r$ all correct processes will observe the outputs of all previously invoked protocols.
Therefore, by the end of round $r+1$ all correct processes that did not apply monitor\!\_halting already, will either be able to apply monitor\!\_halting by the end of that round, or will set $v:=\bot$, since all previously invoked protocols produced output and even stopped.  
Since the latest invoked protocol is guarantee to produce an output of $\bot$, those that have halted will reach the same decision. 
Notice that those processes that do not halt will start another protocol in which every correct process  that invoked it has input $\bot$ and the rest are not participating. By \corollaryref{cor:partial},  by the end of round $r+3$ they will decide the same decision value and will halt by the end of round $r+4.$

\noindent {\bf Case $p$ uses $H_2$}: As in the previous case,
there are three possibilities, one in which $p$ noticed that all previously invoked protocols stopped.  In this case,   \lemmaref{lem:input-agree} implies that all correct processes will observe that all previously invoked protocols reported output by the end of $r+1$ and have stopped by the end of round $r+2$. Some may use rule $H_3$ or rule $H_4$ to  apply monitor\!\_halting and decide the same, and some will invoke the next protocol with input $\bot$ and will reach the same decision by round $r+4$ and will halt by the end of round $r+5.$ 

Otherwise, when $p$ executed round $r$ it noticed that by the end of that round all previous protocols stopped and only the one that started at the beginning of round $r-1$ did not stop yet
and the values of $early$ that $p$ received in round $r-2$ imply that   $|\{q \mid early_q=\true  \mbox{ or $q$ halted}\}| \ge n-t$.  Since no correct process  halted earlier, in round $r-2$ every correct process  sets $v:=\bot.$ The protocol that started in round $r-1$ will produce output of $\bot$ in round $r$ and stop by round $r+1.$ Thus,  every correct process  will reach the same decision and will use rule $H_3$ to halt by the end of round $r+1$.

Otherwise, when $p$ executed round $r-1$ it noticed that by the end of that round all previous protocols stopped and only the one that started at the beginning of round $r-2$ did not stop yet. 
Moreover, $p$ received in round $r-3$,  $|\{q \mid early_q=\true  \mbox{ or $q$ halted}\}| \ge t+1$.  And since no correct process  halted earlier, as in the case for halting rule $H_1$, we are done. 

\noindent {\bf Case $p$ uses $H_3$}: 
Here 
we need to consider 
the case
were all previously invoked protocols were stopped.
In 
this case
every other correct process that did not apply monitor\!\_halting in round $r$ will notice currently running protocols producing outputs by the end of round $r+1$ (\theoremref{thm:agree}, Statement~\ref{t:within2}) and stopping by the end of round $r+2.$  Therefore, by the end of in round $r+1$ every correct process  that will not halt by the end of round $r+1$  will set $v:=\bot$.  Thus, all correct processes participating in the new protocol in round $r+2$ will have an input $\bot$, and every correct process  not participating  will assume to have an input $\bot$. Thus, 
 (\corollaryref{cor:partial}) by the end of round $r+3$ that protocol produces an output, and all decides the same decision value and  halt by the end of round $r+4.$

\noindent {\bf Case $p$ uses $H_4$}: 
Here  we need to consider the case where all previously invoked protocols were stopped, and, in addition,  $p$ observes $|\{q\mid early_q=\true \mbox{ or $q$ halted}\}|\ge n-t$, which leads to halting by the end of round $r$. 
In this case, every other correct process that did not apply monitor\!\_halting in round $r$ will notice all previously invoked protocols  producing outputs by the end of round $r+1$ and stopping by the end of round $r+2$ (\theoremref{thm:agree}, Statement~\ref{t:within2}). The property $|\{q\mid early_q=\true \mbox{ or $q$ halted}\}|\ge n-t$ implies that by the end of round $r+1$ or $r+2$ every correct process will notice $|\{q\mid early_q=\true \mbox{ or $q$ halted}\}|\ge t+1$.   
By the end of round $r+1$ all correct processes that did not halt in round $r$, but noticed that all previously invoked protocols stopped by the end of round $r+1$ will apply monitor\!\_halting in that round.  Those that will notice that all previously invoked protocols, except the one starting in round $r+1$, have stopped, will apply monitor\!\_halting.  The same arguments as for the case of using rule $H_3$, the decision value is identical at all correct processes.

By the end of round $r+2$, all other correct processes, that did not already apply monitor\!\_halting, will either observe that  all previously invoked protocols have stopped  and will  apply monitor\!\_halting, or
 will observe that all previously invoked protocols   except the one starting in round $r+1$ have stopped  and will have the condition that $|\{q\mid early_q=\true \mbox{ or $q$ halted}\}|\ge t+1$ and will  apply monitor\!\_decision by the end of round $r+2$ and will halt by the end of round $r+3$, thus potentially ignoring the output of the last protocol. Again, using previous arguments, all decision values are the same.
\end{proof}
\else 
\fi

Lemma \ref{lem:haltf2} and \ref{lem:haltdecide} complete the correctness part of \theoremref{thm:full}.  
To simplify the polynomial considerations we look at a pipeline of monitors.

\subsection{Monitors Pipeline}\label{sec:monitor-pipeline}
The basic monitor protocol runs a sequence of monitors and tests the number of faults' threshold every 4 rounds (Line~\ref{m:2a}).   This allows the adversary to expose more faults in the following round, and be able to further expand the tree before the threshold is noticed the next time the processes execute Line~\ref{m:2a}. To circumvent this we will run a pipeline of 3 additional sequences of monitors on top of the basic one appearing above. Doing this we obtain that in every
round $r$ one of the 4 monitor sequences will be testing the threshold on the number of faults

Monitor sequence $i,$ for $1\le i\le 4$ begins in round $i$ and invokes protocols every 4 rounds, in every round 
 $r$,   $1<r=i+4k < t-1$, it  invokes  protocol $\mathcal{D}_{t-i-4k}$.  Monitor sequence 1 is the basic monitor sequence defined in the previous subsection. 
Each monitor sequence independently runs the basic monitor protocol ( \figureref{alg:monitor}) every 4 rounds.  In the monitor protocol, the test $\bar{r} = j$, which stands for $\bar{r}\equiv r \pmod 4$ in the basic monitor sequence, is replaced with $\bar{r_i}=j$, which stands for $\bar{r_i}\equiv r+1-i \pmod 4=j$  (naturally only for $r+1-i >0$).  Each of the four monitor sequences decides and halts separately, as in the previous section above.

Notice that protocol $\mathcal{D}_{t}$ is invoked only by the basic sequence ( Sequence 1).  For each of the three other monitor sequences, the decision rule is: decide  \bad, if any invoked protocol (in this sequence) outputs \bad, and $\bot$ otherwise. Observe that Lemma \ref{lem:haltf2} and \ref{lem:haltdecide} hold for each individual sequence.

We now state the global decision and global halting rules:

\begin{definition} [Global Halting]
If any monitor sequence halts with $\bad$, or all 4 monitor sequences halt, the process halts. 
\end{definition}

\begin{definition}
The global\_decision is the output of $\mathcal{D}_{t}$, unless  
any monitor sequence returns \bad, in which case the decision is $\bad.$
\end{definition}

The following are immediate consequences of Lemma \ref{lem:haltf2} and \ref{lem:haltdecide} and the above definitions. 
\begin{corollary}\label{lem:haltf2-glboal}
	If $n>3t$ and there are $f$, $f\le t$, corrupt processes then all correct processes halt by the end of round $\min(t+1,f+2)$.
\end{corollary}

\begin{corollary}\label{lem:haltdecide-global}
	If the first correct process halts in round $r$ on $d$ then every correct process applies 
	global\!\_decision by round $\min \{r+4, f+2 ,t+1\}$, halts by round $\min \{r+5, f+2 ,t+1\}$, and obtains the same decision value.
\end{corollary}


\section{Bounding the size of the tree}\label{sec:tree-size}

Following the approach is \cite{GM98}, we make the following definitions:
\begin{definition}
	A node $\sigma z \in \Sigma$ is \emph{fully corrupt} if there does not exist $p \in G$ and  $\sigma' \sqsupseteq \sigma z$ such that $\sigma' \in \RT_p[|\sigma z|+2].$
\end{definition}

\begin{definition}
	A process $z$ is \emph{becomes fully corrupt at} $i$ if exists a node $\sigma z \in \Sigma$ that is fully corrupt, $|\sigma z|=i$ and for every previous node $|\sigma' z|<i$, node $\sigma' z$ is not fully corrupt.
\end{definition}

The following is immediate from the definitions above.
\begin{claim}
	If process $z$ becomes fully corrupt at $i$ then of all the nodes of $\Sigma$ that end with $z$ only nodes of round $i$ and $i+1$ can be fully corrupt. 
\end{claim}
\begin{proof}
	 By definition of fully corrupt, all correct processes will have $z\in \F$ in round $i+2$. So in that round and later all nodes will put $\bot$ in $\RT$ for $z$. 
\end{proof}

Let $\CT$, the \textit{corrupt tree}, be a dynamic tree structure. $\CT$ is the tree of all fully corrupt nodes (note that due to coloring, the set of fully corrupt nodes is indeed a tree). We denote by $\CT[i]$ the state of $\CT$ at the end of round $i$. By the definition of fully corrupt, at round $i$ we add nodes of length $i-2$ to $\CT$. 


We \textit{label} the nodes in $\CT$ as follows: a node $\sigma z \in \CT$ is a \textit{regular} node if process $z$ becomes fully corrupt at $|\sigma z|$ and  $\sigma z \in \CT$ is a \textit{special} node if process $z$ becomes fully corrupt at $|\sigma z| -1$.

Let $\alpha_i$ denote the distinct number of processes that become fully corrupt at round $i$. For convenience, define $\alpha_0=0$ (this technicality is useful in \lemmaref{lem:two-zero-get-stuck}). Let $A=\alpha_0,\alpha_1,\dots$ be the sequence of counts of 
process that become fully corrupt in a given execution.

Following the approach of \cite{GM98}, we define $waste_i = (\sum_{j\le i} \alpha_i) - i$. So $waste_i$ is the number of processes that became fully corrupt till round $i$ minus $i$ (the round number). The following claim connects $waste_i$ to $\cap_{p \in G} \FA[i+3]_p$ the set of fully detected corrupt processes at round $i+3$.

\begin{claim}\label{waste-imples-FA}
	For any round $4\le r \le t+1$, and any correct process we have $\left|\FA[r]\right| \geq \sum_{j\le r-3} \alpha_i$.
\end{claim}
\begin{proof}
	By the definition of $z$ becoming fully corrupt at $i$, all correct processes will have $z \in \F$ in round $i+2$. Due to the gossiping of $\F$, all correct processes will have $z \in \FA$ in round $i+3$. 
\end{proof}

So if $waste_i \ge 6$ then in round $r=i+3$ we will have $\left(\sum_{j\le i} \alpha_i \right) - i \ge 6$ so by  \lemmaref{waste-imples-FA} for each correct process we have $|\FA[r] |\ge r+3$. In this case all correct processes will start in the associated monitor sequence the next protocol with initial value \bad and the protocol and monitor sequence and global protocol will reach agreement and halt on \bad by round $i+6$ (by \lemmaref{lem:input-agree}).

We will now show that if the adversary maintains a small waste (less than 6 by the argument above, but this will work for any constant) then the $\CT$ tree must remain polynomial sized.

The following key lemma shows that the adversary cannot increase the number of leaves by ``cross contamination". In more detail, if the adversary causes two fully corrupt processes at round $i_1$ followed by a sequence of rounds with exactly one fully corrupt process at each round followed by a round with no fully corrupt process at that round then this action essentially keeps the tree $\CT$ growing at a slow (polynomial) rate.
We note that the focus on ``cross contamination" follows the approach of \cite{GM98}. But they only verify the case of two fully corrupt followed by a round with no fully corrupt. We have identified a larger family of adversary behavior that does not increase the waste (in the long run).  Our proof covers this larger set of behaviors and this requires additional work.

\begin{lemma}\label{lem:two-bad-in-two-rounds-ok}
	Assume $0 < i_1 < i_2 $ such that  $ \alpha_{i_1} = 2 $,  $ \alpha_{i_2}=0$ and   for all $i_1<i<i_2$, $\alpha_i=1$ then  for any $\sigma \in \Sigma_{i_1-1} \cap \CT$ it is not the case that there exists $\sigma p \tau \in \Sigma_{i_2+1} \cap \CT$ and there exists  $\sigma q \tau \in \Sigma_{i_2+1} \cap \CT$ (so there is at most one extension). Moreover the size of the subtree starting from $\sigma p$ or $\sigma q$ and ending in length $i_2+1$ is bounded by $O((i_2 -i_1)^2)$.
\end{lemma}

See the additional analysis in \sectionref{sec:proofs}.

To bound the size of $\CT$, we partition the sequence $A=\alpha_0,\alpha_1,\dots$ by iteratively marking subsequences  using the following procedure. For each subsequence we mark, we prove that it either causes the tree to grow in a controllable manner (so the ending tree is polynomial), or it causes the tree to grow considerably (by a factor of $O(n)$ ) but at the price of increasing the waste by some positive constant. Since the waste is bounded by a constant, the result follows.
 
\beginsmall{enumerate}
	\item By \lemmaref{lem:two-zero-get-stuck} we know that if $A$ contains a $0(1)^*0$ (a sequence starting with 0 then some 1's then 0) then it contains it just once  as a suffix of $A$. Moreover, this suffix does not increase the size of the tree by more than $O(n)$. Let $A_1$ be the resulting unmarked sequence after marking such a suffix (if it exists).
	
	\item Mark all subsequences in $A_1$ of the form $2(1)^*0$ (a sequence starting with 2 then some 1's then 0). By \lemmaref{lem:two-bad-in-two-rounds-ok} each such occurrence  will not increase the number of leafs in $\CT$ (but may add branches that will close whose total size is at most $n^2$ over all such sequences). Let $A_2$ be the remaining unmarked subsequences.
	
	\item Mark all subsequences in $A_2$ of the form $X(1)^*0$ where $X \in \{3,\dots,t\}$ (a sequence starting with 3 or a larger number followed by some 1's then 0). By \lemmaref{lem:bounded-blowup} each occurrence of such a sequence may increase the size of the tree multiplicatively by $O(n)$ leafs and $O(n^2)$ non-leaf nodes, but this also increases the $waste$ by $c-1>1$ (where $c$ is the first element of the subsequence). Observe that the remaining unmarked subsequences do not contain any element that equals 0. Let $A_3$ be the remaining unmarked subsequences.
	
	\item Mark all subsequences of the form $Y(1)^*$ where $Y \in \{2,\dots,t\}$ (a sequence whose first element is 2 or a larger number followed by some 1's but no zero at the end). Again, by \lemmaref{lem:bounded-blowup} each such occurrence may increase the size of the tree by $O(n)$ leafs and $O(n^2)$ non-leafs, but this also increases the $waste$ by $c>1$. Let $A_4$ be the remaining unmarked.

	\item Since $A_3$ contains no element that equals zero and we removed all subsequences that have element of value 2 or larger as the first element then $A_4$ must either be empty or $A_4$ is a prefix of $A$ of the form $(1)^*$ (a series of 1's ). Since it is a prefix of $A$ then a sequence of 1's keeps at most one leaf. So the tree remains small.
\endsmall{enumerate}

Thus, the size of $\CT$ is polynomial, which by \lemmaref{IT-bounded-by-CT} bounds the size of $\IT$.  This completes the proof of  \theoremref{thm:full}.

\subsection{Additional Analysis}\label{sec:proofs}

The following lemma bounds the size of $\IT$ as a function of the size of $\CT$ times $O(n^7)$.
\begin{lemma}\label{IT-bounded-by-CT}
	If $\sigma \in \IT$ and $|\sigma|>7$ then there exists $\sigma' \sqsubset \sigma$ with $|\sigma'| \ge |\sigma|-7$ such that $\sigma' \in \CT$.
\end{lemma}

\ifdefined\LONG

\begin{proof}
	Seeking a contradiction let $\sigma = \sigma' \tau$ be of minimal length such that $\sigma \in IT$, $|\sigma|>7$, $|\tau|=7$ and there does not exist $\sigma' \tau' \in  \CT$ such that $\tau' \sqsubseteq \tau$.
	
	Let $w$ be the first element in $\tau$ so $\sigma' w \sqsubseteq \sigma' \tau$ then since $\sigma' w \notin \CT$ then by definition, some correct process will have $\sigma' w \in \RT[|\sigma' w|+2]$. By \theoremref{thm:main} statement~\ref{c:two-rounds} all correct processes will have $\sigma' w \in \RT[|\sigma'|+5$ and will close the branch $\sigma' w$ by round $|\sigma '|+6$ (see \decayrule) a contradiction to the assumption that $\sigma \in IT$ and $|\tau|=7$.
\end{proof}
\else 
\fi

The following lemma shows that the protocol stops early if the adversary causes two rounds with no new fully corrupt and only one fully corrupt per round between them.

\begin{lemma}\label{lem:two-zero-get-stuck}
	If exists $0 \le i_1 < i_2 $ such that $ \alpha_{i_1}= 0 $, $ \alpha_{i_2}=0$ and for all $i_1<i<i_2$, $\alpha_i=1$ then all processes will halt by the end of round $i_2 +5$. 
\end{lemma}

\ifdefined\LONG

\begin{proof}
	The only fully corrupt process that can appear in round $i_1+1$ is the new one from $\alpha_{i_1+1}=1$ (because $\alpha_{i_1}=0$ and a process can be as a node in $\CT$ for only two rounds starting from the first round it is fully corrupt). A simple induction shows that at round $i_1+j$ only the new fully corrupt node of round $i_1+j$ can appear. Once we reach round $i_2$ then no node can be fully corrupt so all branches will close and all processes will halt by the end of round $i_2 +5$.
\end{proof}
\else 
\fi

We now prove the main technical result of this section  \lemmaref{lem:two-bad-in-two-rounds-ok}. It shows that having two fully corrupt then a series of one fully corrupt then a round with no fully corrupt does not increase the number of leafs in the tree. This can add some non-leaf nodes to the tree, but the overall addition of such nodes is bounded by a multiplicative factor of $O(n^2)$ over all such sequences.

\ifdefined\LONG

\begin{proof}[Proof of Lemma \ref{lem:two-bad-in-two-rounds-ok}]

	Let processes $p,q$ be the two that become fully corrupt at $i_1$. We begin with the case that $i_2=i_1+1$ such that there is no process that becomes fully corrupt at $i_2$. Consider any $\sigma \in \CT$ where $|\sigma|=i_1 -1$. The following is the subtree of $\sigma \in \CT$ that we will analyze:
	
	\Tree [.$\sigma$  [.p  q ].p  [.q  p ].q   ].$\sigma$

	The following analysis for process $p$ shows that either $\sigma p$ or $\sigma q$ or $\sigma qp$ will quickly be in $\RT$. Note that this implies that $p,q$ can extend any node $\sigma \in CT$ into at most one node of length $i_2$ in $\CT$. 
	
	Let {\it correctDetector} be the set of correct processes that detect $\sigma p$ via the \textbf{Not Voter} detection rule in round $|\sigma p| +1$. Let  {\it correctVoter} be the remaining correct processes (that are not in {\it correctDetector}). Note that by definition of \textbf{Not Voter}, the value of all those in {\it correctVoter} must be the same. Let $d$ be this value.
	
	For each $\sigma p u \in\Sigma$ with $u\neq q$ we have that $\sigma p u \notin \CT$ (because $\alpha_{i_1}=0$). So $\sigma p u \in \RT[|\sigma p u|+2]$ for some correct processes and hence their value is fixed (otherwise $\sigma p \in \RT$ and we are done) and all correct processes will have $\sigma pu \in \RT[|\sigma p|+5]$. Let {\it faultyEcho} be the set of corrupt children of $\sigma p$ whose value is fixed to $d$. Let {\it faultyEchoOther} be the remaining corrupt process that are children of $\sigma p$ whose value is fixed to $\ne d$. Note that $\sigma p$ has $n-|\sigma p|$ children of which all but child $\sigma p q$ must be fixed. Hence $|{\it faultyEcho}| + |{\it faultyEchoOther}| \ge n-|\sigma p|-1$.  
	
	There are three cases to consider:	
	
	\textbf{Case 1}: If $|correctVoter|+ |{\it faultyEcho}| \ge n-t -1$ then $\sigma p \in \RT[|\sigma p|+5]$ for all correct processes since all these $n-t-1$ children of $\sigma p$ will appear in $\RT[|\sigma p|+5]$ and so $\sigma p\in \RT[|\sigma p|+5]$  using the \rgcrule.

	Otherwise, $|correctVoter|+ |{\it faultyEcho}| \le n-t -2$ so it must be that $|correctDetector| + |{\it faultyEcho}Other|  \ge t+1 - |\sigma p| = t+2 - |\sigma pq|$. This is because  $\sigma p$ has $n-|\sigma p|$ children and each one of them except of child $q$ must fix their value in $\RT[|\sigma p|+3]$.
	
	\textbf{Case 2}: If $|correctDetector|\ge t+2 - |\sigma p q| $ then \srule will fire on the level $i_1+1$ node $\sigma q p$. This will occur because all other children of $\sigma q$ are not fully corrupt - hence will appear in $\RT [|\sigma q|+5]$. The only case in which \srule may not fire is if in the meantime $\sigma q \in \RT$ in which case we are done.
	
	\textbf{Case 3}: It must be that ${\it correctDetector} \le t$, hence ${\it correctVoter} \ge t+1$ on value $d$ (because $\sigma$ contains no correct process). Since $correctVoter \ge t+1$ then all correct processes  will see that $\sigma p$ is \textit{leaning towards} $d$ (see definitions 3. and 4. in the fault detection rules).
	
	For any $w \in {\it faultyEchoOther}$, since $w$ does not become fully corrupt at $i_1$ or $i_1+1$ it must be that are at least $t+1$ correct processes that are children of $\sigma p w$ that hear from $\sigma p w$ a value $d'$, $d' \ne d$. So the conditions of \textbf{Not Masking} for $\sigma p w$  hold.

	This implies that $w$ is `forced' to send $\bot$ for $\sigma qp w$ to all correct processes. For if $w$ sends $d' \neq \bot$ to any correct processes for $\sigma qp w$ then by  \textbf{Not Masking} rule at round $|\sigma p w|+2 = |\sigma qpw| +1$ these correct processes will detect $w$ as corrupt and in the same round mask $\sigma qpw$ to $\bot$.
	
	Therefore there will be $|correctDetector| + |{\it faultyEcho}Other|  \ge t+2 - |\sigma qp|$ children of $\sigma qp$ that will appear in $\IT$ with value $\bot$ and since there is no process that becomes fully corrupt at $i_1+1=i_2$ then all other children of $\sigma q$ must appear in $\RT [|\sigma q|+5]$. So the \srule will fire on the level $i_1+1$ node $\sigma q p$. This completes the proof for the case $i_1-i_2=1$.

	We can now consider the case where $i_1-i_2>1$. The key observation is that the above argument required two properties for a process $z$ that becomes fully corrupt at round $i$. The first is that all the level $i$ nodes of the form $\sigma z$ have all their children (except one) fixed to some value. The second is that the level $i+1$ nodes of the form $\sigma' z$ have the property that all other children of $\sigma '$ are fixed. 
	
	Intuitively, if a child $\sigma z u$ is fixed to the majority value of $\sigma z$ then $\sigma z u $ will help fix $\sigma z$ using the relaxed rule. Otherwise, $\sigma zu$ is fixed to some $d'$, which implies that at least $t+1$ correct processes received $d'$ from $\sigma zu$. Hence $\sigma z u$ must be a masker for the round $i+1$ node $\sigma' z$.
	
	Next we observe the structure of $\CT$ given a sequence with $i_1-i_2>1$. Let $p,q$ be the two processes in $i_1$, let $\ell=i_2-i_1+1$ and denote by $x_3,\dots,x_{\ell}$ the remaining fully corrupt by order of appearance.  Using an inductive argument one can show that any $\CT$ graph will be a subgraph of the following: for every node $\sigma \in \CT$ of length $i_1-1$ there will be two branches that we call \textit{special branches}. These branches will be $\sigma p q x_3 \dots x_\ell$ and $\sigma q p x_3 \dots x_\ell$. Observe that these branches contain only special nodes.  In addition, there will be regular branches as follows: $\sigma p x_3 \dots x_\ell$, $\sigma q x_3 \dots x_\ell$, $\sigma p q x_4 \dots x_\ell$, $\sigma q p x_4 \dots x_\ell$, \dots $\sigma p q x_3 \dots  x_i x_{i+2} \dots x_\ell$, $\sigma q p x_3 \dots x_i x_{i+1} \dots x_\ell$, $\dots$, $\sigma p q x_3 \dots x_{\ell -2} x_\ell$, $\sigma q p x_3 \dots x_{\ell-2}, x_\ell$. Observe that all these regular branches contain regular nodes and that all their children will be fixed due to round $i_2$ having no fully corrupt process. The number of regular branches is $O(i_2 -i_1)$ and the length of each branch is bounded by $O(i_2 -i_1)$.

	\Tree [.$\sigma$  [.p  [.q  [.$x_3$ [.$x_4$ $x_5$ ].$x_4$  $x_5$  ].$x_3$ [.$x_4$ $x_5$ ].$x_4$ ].q  [.$x_3$ [.$x_4$ $x_5$ ].$x_4$ ].$x_3$ ].p  [.q  [.$x_3$ [.$x_4$ $x_5$ ].$x_4$ ].$x_3$ [.p  [.$x_4$ $x_5$ ].$x_4$ [.$x_3$ $x_5$ [.$x_4$ $x_5$ ].$x_4$ ].$x_3$  ].p  ].q   ].$\sigma$
	
	The above tree is an example for $i_2-i_1=4$. The two special branches are the rightmost and leftmost paths. All other leafs are the endpoints of all the regular branches. Observe that given one more fully corrupt, each special branch is split into two branches, one extends the original special branch and the other is a new regular branch that continues as a path. Also observe that one more fully corrupt will simply extend the path of each regular branch by one.
	
	As all the regular branches will have all their children fixed, they cannot be used as leafs to extend the tree. Since there are $O(i_2-i_1)$ regular branches and each of them is of length at most $O(i_2-i_1)$ then the total amount of nodes added in this process is $O((i_2-i_1)^2)$ per each leaf in $\CT$ of length $i_1-2$. So if the size of the tree without this subtree is $x$ then the total number of non-tree nodes added by these types of sequences is at most $O(x n^2)$ (this is a crude bound that can be improved).
	
	We now need to show that at least one of the special branches gets fixed.	
	Since all the regular branches cannot expand, our goal is to prove that it cannot be the case that both  special branches are not fixed (in the $i_1-i_2=1$ the analogue is that either $\sigma pq$ or $\sigma qp$ is fixed). Given the key observation and the structure statement we can now apply a similar argument as we did for $p$ in the $i_1-i_2=1$ case. We start with $x_\ell$ and going towards $p,q$. We will show that in each iteration on level $i$ we either fix one of the special branches (and we are done) or we have sufficient conditions to use main argument on level $i-1$.
	
	For the base case, consider $x_\ell$. Because $i_2=0$ then all the level $i_1+\ell-2$ nodes of the form $\sigma' x_\ell$ (for any $\sigma'$) have all their children fixed. So we can apply the main argument: if all these level $i_1+\ell-2$ nodes get fixed using the \rgcrule then all the regular branches ending with $x_{\ell-1}$ have all their children fixed and the two special branches ending $x_{\ell-1}$ each have their parent with $x_{\ell-1}$ as a only child. Therefore we continue by induction. 
	Otherwise, by the argument above, all the level $i_1+\ell-2 +1$ nodes of the form  $\sigma' x_\ell$ (for any $\sigma'$) will be fixed \srule. In particular this includes the special branch. So we are done.
	
	For the general case, we assume that all level $i_1+j-2$ nodes of the form $\sigma' x_j$ (for any $\sigma '$) have all their children fixed and that for the two special branches, the parents of $x_j$ have $x_j$ as their only  child. 
	Again we can apply the $i_i-i_2=1$ arguments: If all these level $i_1+j-2$ node get fixed using the \rgcrule then we continue by induction to $j-1$. Otherwise, by the argument above, all the level $i_1+j-2 +1$ nodes of the form  $\sigma' x_\ell$ (for any $\sigma'$) will be fixed by the \srule. In particular this includes the special branch. 
	So we are done since the special branch is fixed	
\end{proof}
\else 
\fi

The following lemma shows that having a large number (3 or more) of processes becoming fully corrupt at a given round, followed by a sequence of 1's and then maybe followed by 0 does increase the number of leafs considerably. Note that if $\alpha_{i_1-1} + \alpha_{i_1} \ge 6$ then the monitor process will cause the protocol to reach agreement and stop in a constant number of rounds. So we only look at the case that $\alpha_{i_1-1} + \alpha_{i_1} < 6$.

\begin{lemma}\label{lem:bounded-blowup}
	If $2 <\alpha_{i_1}$, $\alpha_{i_1-1} + \alpha_{i_1} < 6$,  $ \alpha_{i_2} \in \{0,1\}$ and   for all $i_1<i<i_2$, $\alpha_i=1$ then  for any $\sigma \in \Sigma_{i_1-1} \cap \CT$ there are at most $O(i_2-i_1)$ nodes of the form $\sigma \tau \in \Sigma_{i_2+1} \cap \CT$. Moreover the size of the subtree starting from $\sigma$ and ending in length $i_2+1$ is bounded by $O ( (i_2 -i_1)^2)$.
\end{lemma}

\ifdefined\LONG
\begin{proof}
	Using an overly  pessimistic argument, every node $\sigma \in \Sigma_{i_1-1} \cap \CT$ can have at most $\alpha_{i_1-1} \cdot \alpha_{i_1} \le 16 =O(1)$ nodes of length $i_1+2$ in $\CT$. Even if each such node is a special node then after $O(i_2-i_1)$ rounds of just one fully corrupt each round, each such node of length $i_1+2$ will generate at most $O(i_2-i_1)$ regular branches, each is a path with at most $O(i_2-i_1)$ nodes.	


\end{proof}		
\else 
\fi

\section{Conclusion}
In this paper we resolve the problem of the existence of a protocol with polynomial complexity and optimal early stopping and resilience. The main remaining open question is reducing the complexity of such protocols to a low degree polynomial. Another interesting open problem is obtaining unbeatable protocols \cite{CGM14} (which is a stronger notion than early stopping).

We would like to thank Yoram Moses and Juan Garay for insightful discussions and comments.


\bibliographystyle{alpha}

\begin{thebibliography}{BNDDS92}
	
	\bibitem[BG93]{BG93-votes}
	Piotr Berman and Juan~A. Garay.
	\newblock Cloture votes: n/4-resilient distributed consensus in t+1 rounds.
	\newblock {\em Mathematical Systems Theory}, 26(1):3--19, 1993.
	
	\bibitem[BGP92]{BGP92}
	Piotr Berman, Juan~A. Garay, and Kenneth~J. Perry.
	\newblock Optimal early stopping in distributed consensus.
	\newblock In Adrian Segall and Shmuel Zaks, editors, {\em Distributed
		Algorithms}, volume 647 of {\em Lecture Notes in Computer Science}, pages
	221--237. Springer Berlin Heidelberg, 1992.
	
	\bibitem[BNDDS92]{BDDR92}
	Amotz Bar-Noy, Danny Dolev, Cynthia Dwork, and H.~Raymond Strong.
	\newblock Shifting gears: changing algorithms on the fly to expedite byzantine
	agreement.
	\newblock {\em Inf. Comput.}, 97:205--233, April 1992.
	
	\bibitem[CGM14]{CGM14}
	Armando Casta{\~{n}}eda, Yannai~A. Gonczarowski, and Yoram Moses.
	\newblock Unbeatable consensus.
	\newblock In Fabian Kuhn, editor, {\em Distributed Computing - 28th
		International Symposium, {DISC} 2014, Austin, TX, USA, October 12-15, 2014.
		Proceedings}, volume 8784 of {\em Lecture Notes in Computer Science}, pages
	91--106. Springer, 2014.
	
	\bibitem[CL99]{CCL99}
	Miguel Castro and Barbara Liskov.
	\newblock Practical byzantine fault tolerance.
	\newblock In {\em Proceedings of the third symposium on Operating systems
		design and implementation}, OSDI '99, pages 173--186, Berkeley, CA, USA,
	1999. USENIX Association.
	
	\bibitem[DRS90]{DRS90}
	Danny Dolev, Ruediger Reischuk, and H.~Raymond Strong.
	\newblock Early stopping in byzantine agreement.
	\newblock {\em J. ACM}, 37:720--741, October 1990.
	
	\bibitem[DS82]{DS82}
	Danny Dolev and H.~Raymond Strong.
	\newblock Polynomial algorithms for multiple processor agreement.
	\newblock In {\em ACM Symposium on Theory of Computing}, pages 401--407, New
	York, NY, USA, 1982. ACM.
	
	\bibitem[FL82]{FL82}
	Michael~J. Fischer and Nancy~A. Lynch.
	\newblock A lower bound for the time to assure interactive consistency.
	\newblock {\em Inf. Process. Lett.}, 14(4):183--186, 1982.
	
	\bibitem[FM88]{FM88}
	Paul Feldman and Silvio Micali.
	\newblock Optimal algorithms for byzantine agreement.
	\newblock In {\em ACM Symposium on Theory of Computing}, pages 148--161, 1988.
	
	\bibitem[FM97]{FM97}
	Pesech Feldman and Silvio Micali.
	\newblock An optimal probabilistic protocol for synchronous byzantine
	agreement.
	\newblock {\em SIAM J. Comput.}, 26(4):873--933, 1997.
	
	\bibitem[GM93]{GM93}
	Juan~A. Garay and Yoram Moses.
	\newblock Fully polynomial byzantine agreement in t + 1 rounds.
	\newblock In {\em Proceedings of the twenty-fifth annual ACM symposium on
		Theory of computing}, STOC '93, pages 31--41, New York, NY, USA, 1993. ACM.
	
	\bibitem[GM98]{GM98}
	Juan~A. Garay and Yoram Moses.
	\newblock Fully polynomial byzantine agreement for processors in rounds.
	\newblock {\em SIAM J. Comput.}, 27:247--290, February 1998.
	
	\bibitem[KAD{\etalchar{+}}07]{ZYZ07}
	Ramakrishna Kotla, Lorenzo Alvisi, Mike Dahlin, Allen Clement, and Edmund Wong.
	\newblock Zyzzyva: speculative byzantine fault tolerance.
	\newblock In {\em Proceedings of twenty-first ACM SIGOPS symposium on Operating
		systems principles}, SOSP '07, pages 45--58, New York, NY, USA, 2007. ACM.
	
	\bibitem[KM13]{KM03}
	Dariusz~R. Kowalski and Achour Most{\'e}faoui.
	\newblock Synchronous byzantine agreement with nearly a cubic number of
	communication bits: Synchronous byzantine agreement with nearly a cubic
	number of communication bits.
	\newblock In {\em Proceedings of the 2013 ACM Symposium on Principles of
		Distributed Computing}, PODC '13, pages 84--91, New York, NY, USA, 2013. ACM.
	
	\bibitem[LSP82]{LPS82}
	Leslie Lamport, Robert Shostak, and Marshall Pease.
	\newblock The byzantine generals problem.
	\newblock {\em ACM Trans. Program. Lang. Syst.}, 4:382--401, July 1982.
	
	\bibitem[PSL80]{PSL80}
	Marshall Pease, Robert Shostak, and Leslie Lamport.
	\newblock Reaching agreement in the presence of faults.
	\newblock {\em J. ACM}, 27(2):228--234, 1980.
	
\end{thebibliography}
\newcommand{\etalchar}[1]{$^{#1}$}

\end{document}